\documentclass[10pt, journal, letter, twocolumn]{IEEEtran}

\usepackage{pgfplots}

\pgfplotsset{compat=1.10}

\usepgfplotslibrary{fillbetween}

\usepackage{lipsum}
\usepackage{makeidx}
\usepackage{enumerate}
\usepackage{color}
\usepackage{cite}
\usepackage{amsmath,amsthm}   
\usepackage{amssymb}
\usepackage{nomencl}
\usepackage{multirow}
\usepackage{graphicx}
\usepackage{multirow}
\usepackage{anysize}
\usepackage{float}
\usepackage{epstopdf}
\usepackage{threeparttable}
\usepackage{multicol}
\DeclareGraphicsExtensions{.pdf,.jpeg,.png,.jpg,.emf,.eps}
\usepackage{geometry}

\usepackage{calc}
\usepackage{here}
\usepackage{url}
\usepackage{upref}
\usepackage{comment}
\usepackage{times}
\usepackage{dsfont}
\usepackage{epic,eepic}
\usepackage{rawfonts}
\usepackage[T1]{fontenc}
\usepackage{latexsym}
\usepackage{amsfonts}

\usepackage{capitalgreekitalic}

\usepackage{geometry}
\usepackage[font=small,labelfont=bf]{caption}
\usepackage[font=small,labelfont=bf]{subcaption}

\usepackage{tikz,pgfplots}
\usetikzlibrary{calc}
\usetikzlibrary{intersections}

\newtheorem{theorem}{Theorem}

\newtheorem{lemma}{Lemma}

\theoremstyle{definition}
\newtheorem{definition}{Definition}

\begin{document}

\title{Improper Signaling for SISO Two-user Interference Channels with Additive Asymmetric Hardware Distortion}

\author{Mohammad Soleymani$^*$, \emph{Student Member, IEEE}, Christian Lameiro$^*$, \emph{Member, IEEE}, Ignacio Santamaria$^\dag$ \emph{Senior Member, IEEE} and Peter J. Schreier$^*$, \emph{Senior Member, IEEE}
\\ \thanks{
$^*$Mohammad Soleymani, Christian Lameiro and Peter J. Schreier are with the Signal and System Theory Group, Universit\"at Paderborn, Germany, http://sst.upb.de  (emails: \protect\url{{mohammad.soleymani, christian.lameiro, peter.schreier}@sst.upb.de}).  

$^\dag$Ignacio Santamaria is with the Department of Communications Engineering, University of Cantabria (email: \protect\url{i.santamaria@unican.es}).

The work of M. Soleymani, C. Lameiro and P. J. Schreier was supported by the German Research Foundation
(DFG) under grants LA 4107/1-1, SCHR 1384/7-1 and SCHR 1384/8-1. 
The work of I. Santamaria was supported by MINECO of Spain and AEI/FEDER funds of the E.U., under grant TEC2016-75067-C4-4-R (CARMEN).}
}
\maketitle

\begin{abstract}
Hardware non-idealities are among the main performance restrictions  for  upcoming wireless communication systems. Asymmetric hardware distortions (HWD) happen when the impairments of the I/Q branches are correlated or imbalanced, which in turn generate improper additive interference at the receiver side. When the  interference is improper, as well as in other interference-limited scenarios, improper Gaussian signaling (IGS) has been shown to provide rate and/or power efficiency benefits. 
 In this paper, we investigate the rate benefits of IGS in a two-user interference channel (IC) with additive asymmetric HWD when interference is treated as noise.  
We propose two  iterative algorithms to optimize the parameters of the improper transmit signals. We first rewrite the rate region as an pseudo-signal-to-interference-plus-noise-ratio (PSINR) region and employ  majorization minimization and fractional programming to find a suboptimal solution for the achievable user rates. Then, we propose a simplified algorithm based on a separate optimization of the powers and complementary variances of the users, which exhibits lower computational complexity. We show that IGS can improve the performance of the two-user IC with additive HWD. Our proposed algorithms outperform  proper Gaussian signaling and competing IGS algorithms in the literature that do not consider asymmetric HWD.
\end{abstract} 

\section{Introduction}

One of the targets of 5G is reaching a data rate more than 1000 times greater than the data rate of current cellular systems \cite{andrews2014will}. 
However, reaching this goal entails many challenges. 
Among them is to overcome the non-idealities, i.e., hardware distortions (HWD), of devices which can result in a substantial performance degradation \cite{ buzzi2016survey, heath2016overview, javed2017asymmetric}. 
HWD are due to various imperfections in transceivers, including I/Q imbalance, non-linear power amplifiers, imperfect and/or low resolution analog-to-digital
and digital-to-analog converters, frequency/phase offset and so on \cite{heath2016overview,  javed2017asymmetric, zhu2017analysis, zhang2016achievable, xia2015hardware,bjornson2014massive,bjornson2013capacity,bjornson2013new}. 
Another main challenge for data-rate enhancement is to handle interference from other users, and hence interference management techniques play a key role in 5G \cite{andrews2014will}. 
Recently, it has been shown that improper Gaussian signaling (IGS) can improve the performance of   
various {\em interference-limited} systems \cite{ javed2017full,cadambe2010interference,yang2014interference, ho2012improper, zeng2013transmit, nguyen2015improper, lagen2016superiority, kurniawan2015improper, lameiro2017rate, lameiro2015benefits,amin2017overlay,soleymani2019energy,soleymani2019robust,Sole1909:Energy,lagen2016coexisting,gaafar2017underlay}. 
In IGS  schemes, the real and imaginary parts of the signal are correlated and/or have unequal powers \cite{schreier2010statistical,adali2011complex}. 
While proper Gaussian signaling (PGS) achieves channel capacity for point-to-point communications in the presence of proper Gaussian noise \cite{cover2012elements}, 
this is not the case 
under improper Gaussian noise that arises as a result of asymmetric HWD \cite{javed2017asymmetric, javed2017full,javed2018improper,javed2017impact}. 

\subsection{Related work}
The effect of HWD  is studied in \cite{zhu2017analysis, zhang2016achievable, xia2015hardware,bjornson2014massive,bjornson2013capacity,bjornson2013new} for various scenarios. 
In  \cite{zhu2017analysis}, the secrecy performance of
downlink massive multiple-input multiple-output (MIMO) systems was considered with HWD and a passive
multiple-antenna eavesdropper. 
The paper \cite{zhang2016achievable} analyzed the achievable rate of massive MIMO
 systems with Rician channels and HWD. 
In \cite{xia2015hardware}, the authors considered a full-duplex massive MIMO relay with HWD and proposed a scheme to mitigate the distortion by exploiting statistical
knowledge of the channels. 
In \cite{bjornson2014massive}, the authors studied a massive MIMO system with a new system model for HWD at the transceivers.
The paper \cite{bjornson2013new} studied the performance of dual-hop relaying with different protocols in the presence of HWD.

In the aforementioned papers, symmetric HWD are considered. 
Nevertheless, HWD can, in general, provoke asymmetric or improper  distortion in both the transmitted and received signal
\cite{javed2018improper,boulogeorgos2016energy,javed2017asymmetric, javed2017full,javed2017impact}. 
The paper \cite{javed2017asymmetric} considered  IGS in  a single-input, multiple-output (SIMO) system with additive asymmetric HWD and showed that IGS improved the performance of the system.
In \cite{javed2017full}, the authors investigated the effect of IGS in a relay network with additive asymmetric HWD. 
They maximized the achievable rate of the relay network by optimizing the complementary variance of the transmitted signal in the source and relay nodes.

Improper signaling schemes have also been proposed to improve different performance metrics in interference-limited networks with ideal devices \cite{cadambe2010interference,yang2014interference, ho2012improper, zeng2013transmit,nguyen2015improper, lagen2016superiority, kurniawan2015improper, lameiro2017rate, lameiro2015benefits, amin2017overlay,soleymani2019energy,soleymani2019robust,gaafar2017underlay}. 
In \cite{cadambe2010interference}, IGS was considered as an interference management tool for the first time in the literature, where the authors considered a three-user interference channel (IC) and showed that 
IGS can improve the degrees-of-freedom (DoF) in this scenario.
The paper \cite{yang2014interference} showed that IGS can increase the DoF of MIMO X channels. 
In \cite{ho2012improper, zeng2013transmit, nguyen2015improper, lagen2016superiority, kurniawan2015improper, lameiro2017rate, lameiro2015benefits,amin2017overlay,soleymani2019energy,soleymani2019robust,lagen2016coexisting},
the authors studied the performance of IGS when Treating Interference as Noise (TIN) was the strategy used for decoding. 
 The paper \cite{ho2012improper} showed that IGS can improve the performance of the two-user interference channel, while in \cite{zeng2013transmit} IGS was used to 
 optimize the rate of the $K$-user MIMO interference channel. 
Moreover, the authors in \cite{zeng2013transmit} derived the rate region of the two-user single-input, single-output (SISO) IC with TIN by solving a semidefinite programming (SDP) problem, showing that  IGS can enlarge the rate region and improve the performance of the system.
The paper \cite{nguyen2015improper} showed that IGS can reduce the symbol error rate of the $K$-user IC. 
In \cite{ lagen2016superiority, kurniawan2015improper, lameiro2017rate},  benefits of IGS were  studied in different Z-IC scenarios. 
In \cite{lameiro2015benefits,amin2017overlay}, the authors showed that IGS improves the performance of underlay and overlay cognitive radio systems, respectively.
Finally,  \cite{gaafar2017underlay} showed that IGS can improve the performance of full-duplex relaying systems with fading channels.

\subsection{Contribution}
In this paper, 
we study the performance of IGS in a two-user IC with additive asymmetric HWD with TIN. 
To the best of our knowledge, it is the first work addressing the SISO IC with asymmetric HWD.
We assume that the transceivers of both users produce additive asymmetric HWD noise, and model the HWD as an additive improper Gaussian noise, similar to \cite{javed2018improper, javed2017asymmetric, javed2017full, javed2017impact}. 
We devise
two  iterative algorithms  to 
derive suboptimal solutions for the achievable rate region of the two-user IC.
To this end, we rewrite the rate region as a pseudo-signal-to-interference-plus-noise-ratio (PSINR) region and employ sequential optimization approaches to solve the resulting problems. 

In our first proposed algorithm, we employ majorization minimization (MM) as well as fractional programming (FP) and the well-known generalized Dinkelbach algorithm. 
MM is an iterative algorithm and  consists of two steps in every iteration: i) majorization, and ii) minimization \cite{sun2017majorization}. In the majorization step, the objective function is approximated by a surrogate function. Then, the approximated problem is solved in the minimization step. 
In other words, MM solves a non-convex optimization problem  by solving a sequence of surrogate optimization problems, which can be solved easier than the original problem \cite{sun2017majorization}. 
In our algorithm, to solve each surrogate problem, we apply the generalized Dinkelbach algorithm,  
which is a powerful tool to solve multiple ratio maximin problems \cite{crouzeix1991algorithms, zappone2015energy}.
In Dinkelbach-based algorithms, an iterative optimization is performed, in which
the fractional functions are replaced by surrogate functions at each iteration. 
The generalized Dinkelbach algorithm permits solving fractional programming efficiently and results in the global optimal solution of the original optimization problem if the optimization problem at each iteration is perfectly solved, i.e., its global optimum is obtained \cite{crouzeix1991algorithms, zappone2015energy, dinkelbach1967nonlinear, shen2018fractional}. 

In our second  proposed algorithm, we employ a separate optimization of powers and complementary variances. 
We first optimize the powers transmitted by the users by employing the well-known bisection method, which transforms the original problem into a sequence of feasibility problems, and derive 
a closed-form solution for the feasibility problem. 
In order to obtain the complementary variances, we employ difference of convex programming (DCP), which is a special case of sequential convex programming (SCP) and falls into MM \cite{sun2017majorization,lanckriet2009convergence}. 
In DCP, the objective function and/or constraints are difference of two convex/concave functions. 
DCP solves a non-convex problem by solving a sequence of convex optimization problems and converges to a stationary point\footnote{A stationary point of a constrained optimization problem satisfies the corresponding
Karush-Kuhn-Tucker (KKT) conditions \cite{lanckriet2009convergence}.} of the original problem \cite{lanckriet2009convergence}.

The main contributions of this paper are as in the following:
\begin{itemize}
\item We first propose an iterative algorithm based on a sequential optimization method, in which we solve a sequence of fractional optimization problems \cite{sun2017majorization, aubry2018new}. 
We derive the global optimal solution of each surrogate problem by FP and the generalized Dinkelbach algorithm.  
Our first proposed algorithm obtains a stationary point of the PSINR region. 

\item We also propose a simplified algorithm that is computationally less expensive  than our proposed algorithm with FP. 
This simplified algorithm is based on a separate optimization of powers and complementary variances of users. We employ a bisection method to obtain the powers and derive a closed-form solution for powers in each iteration.
Then, we employ DCP to find the complementary variances.

\item Our results show that IGS enlarges the achievable rate of the two-user IC in the presence of additive asymmetric HWD, and that
there is a significant performance improvement by IGS for highly asymmetric 
HWD noise. 
Moreover, both of our proposed algorithms outperform existing PGS and other existing IGS algorithms. 
\end{itemize}
\subsection{Paper outline}
The rest of this paper is organized as follows.  Section \ref{secII} describes the scenario and formulates the achievable-rate-region problem. 
 In Section \ref{FP}, we propose our algorithm based on MM and FP, and 
 in Section IV, we develop a simplified version of this algorithm.  
Finally, Section \ref{secIV} presents some numerical results.
 
\section{System Model}\label{secII}
\subsection{Preliminaries of IGS}
Let us consider a zero-mean complex Gaussian random variable $x$ with variance $p_x=\mathbb{E}\{|x|^2\}$ and complementary variance $q_x=\mathbb{E}\{x^2\}$ \cite{schreier2010statistical,adali2011complex}.  
Note that the complementary variance is complex and $|q_x|\leq p_x$. We denote the probability distribution of $x$ by $x\sim\mathcal{CN}(m_{x},p_x,q_x)$, where $m_{x}=0$ is the mean of $x$. 
We define  
the complex correlation coefficient of $x$ as $\tilde{\kappa}_x=\frac{q_x}{p_x}$, where $|\tilde{\kappa}_x|\in [0,1]$ is the so-called circularity coefficient. If $\tilde{\kappa}_x=0$, $x$ is proper; otherwise, it is improper \cite{schreier2010statistical,adali2011complex}. 
We call $x$ maximally improper if $|\tilde{\kappa}_x|=1$.

\subsection{Hardware distortion model}
  \begin{figure}[t]
\centering
\includegraphics[width=0.5\textwidth]{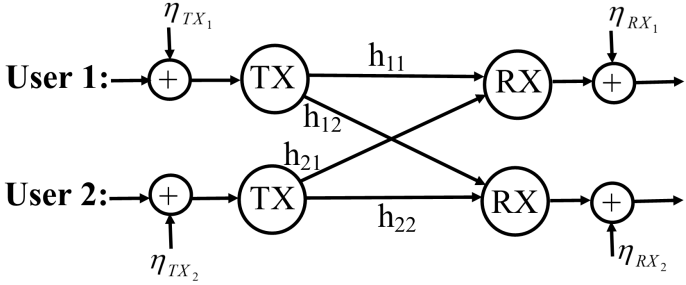}
\caption{The channel model for the SISO two-user IC.}
\label{Fig1}
\end{figure}

In this paper, we employ the distortion model in  \cite{javed2018improper,javed2017full,javed2017asymmetric,javed2017impact} 
and model the 
aggregated  effect of HWD on the transceiver of a communication link with an improper Gaussian additive noise as
 \begin{equation}
y=\sqrt{P}h(x+\eta)+n,
\end{equation} 
where $y$, $x$, $P$, $h$, $\eta$, and $n$ are the received signal, transmitted symbol, transmission power, channel coefficient, aggregated HWD noise and additive complex proper Gaussian noise, respectively. 
The aggregated HWD noise is modeled as an improper complex Gaussian random variable with  probability distribution $\eta\sim \mathcal{CN}(0,\sigma_{\eta}^2,\tilde{\sigma}_{\eta}^2)$, 
where $\sigma_{\eta}^2=\sigma_{\eta_{\text{TX}}}^2+\sigma_{\eta_{\text{RX}}}^2$ and $\tilde{\sigma}_{\eta}^2=\tilde{\sigma}_{\eta_{\text{TX}}}^2+\tilde{\sigma}_{\eta_{\text{RX}}}^2$ are the variance and complementary variance of $\eta$, respectively, 
both of which are composed of contributions at the transmitter side (denoted TX) and the receiver side (denoted RX). Please refer to \cite[Lemma 1]{javed2018improper} for more details about this model.

It is worth mentioning that this model is an extension of the model in \cite{zhu2017analysis, zhang2016achievable, xia2015hardware,bjornson2014massive,bjornson2013capacity,bjornson2013new}, where the HWD is modeled as  additive proper Gaussian noise. 
However, as indicated in e.g., \cite{javed2018improper,javed2017full,schenk2008rf,javed2017asymmetric,javed2017impact,javed2019multiple,javed2019asymmetric,boulogeorgos2016energy},  the aggregated HWD is, in general, improper due to I/Q imbalance.
Note that the variances and complementary variances of HWD noise are not only a function of device parameters, but also a linear function of the transmission power and channel gain, meaning that higher transmission power results in higher HWD noise \cite{javed2017asymmetric,javed2018improper}.
Moreover, even if the channel noise is proper, the aggregated distortion is improper due to the asymmetric HWD. 

\subsection{Network scenario and signal model}
We consider a two-user IC with additive asymmetric HWD at the transmitters and receivers of both users, as depicted in Fig. \ref{Fig1}.\footnote{It is worth mentioning that our algorithms can easily be extended to the $K$-user IC. However, we consider only the 2-user IC for the ease of illustration.} 
We assume that users are allowed to employ IGS and treat the interference as noise.   
Using the  proposed HWD model, the received signals at receiver $k$ is
\begin{equation}
y_k\!=\sqrt{p_1}h_{1k}(x_1+\eta_{1k})+\sqrt{p_2}h_{2k}(x_2+\eta_{2k})+n_k,
\end{equation}
respectively, where 
$x_k$, $h_{jk}$, $n_k$, and $\eta_{jk}$ for $j,k\in\{1,2\}$ are the transmit signal of user $k$, channel 
between transmitter $j$ and receiver $k$, independent zero-mean proper complex Gaussian noise with variance $\sigma^2$, and the  aggregated HWD noise of the link between transmitter $j$ and receiver $k$, 
respectively. 
Since the transmitted signals $x_1$ and $x_2$ are improper complex Gaussian, the achievable rate of user $k\in\{1,2\}$ is \cite{javed2017full, lameiro2017rate,zeng2013transmit} given by \eqref{R1}, shown at the top of the next page,
\begin{figure*}
\begin{equation}\label{R1}
R_k\!\!=\! \frac{1}{2}\!\log_2\!\!\left(\!\!\frac{\left(\sigma^2+\sum_{j=1}^{2}p_j|h_{jk}|^2(1\!+\!\sigma_{\eta_{jk}}^2)\!\right)^2\!\!\!-\!\left|\sum_{j=1}^{2}(q_j\!+\! p_j\tilde{\sigma}_{\eta_{jk}}^2)h_{jk}^2\right|^2}{\left(\sigma^2+\sum_{j=1}^{2}p_j|h_{jk}|^2(1\!+\!\sigma_{\eta_{jk}}^2)\!-\!p_k|h_{kk}|^2\!\right)^2\!\!\!-\!\left|\sum_{j=1}^{2}(q_j\!+\! p_j\tilde{\sigma}_{\eta_{jk}}^2)h_{jk}^2-q_kh_{kk}^2\right|^2}\!\right)\!\!,
\end{equation}
\setcounter{equation}{4}
\begin{align}
\mathbf{a}_k&=\!\left[\!\!\begin{array}{cc}|h_{1k}|^2(1\!+\!\sigma^2_{\eta_{1k}})& |h_{2k}|^2(1\!+\!\sigma^2_{\eta_{2k}})\end{array}\!\!\right]^T,           &  \mathbf{f}_k&=\!\left[\!\!\begin{array}{cc}h_{1k}^2& h_{2k}^2\end{array}\!\!\right]^H,  \label{matf}    \\
\label{matf-1}\tilde{\mathbf{f}}_k&=\!\left[\!\!\begin{array}{cc}h_{1k}^2\tilde{\sigma}^2_{\eta_{1k}}& h_{2k}^2\tilde{\sigma}^2_{\eta_{2k}}\end{array}\!\!\right]^H,      &   \\
\mathbf{b}_1&=\!\left[\!\!\begin{array}{cc}|h_{11}|^2\!\sigma^2_{\eta_{11}}& |h_{21}|^2(1\!+\!\sigma^2_{\eta_{21}})\end{array}\!\!\right]^T,         &  \mathbf{g}_1&=[\begin{array}{ccc}0& h_{21}^2\end{array}]^H,   \\
\mathbf{b}_2&=\!\left[\!\!\begin{array}{cc}|h_{12}|^2(1\!+\!\sigma^2_{\eta_{12}})& |h_{22}|^2\sigma^2_{\eta_{22}}\end{array}\!\!\right]^T,   &  \mathbf{g}_2&=[\begin{array}{ccc}h_{12}^2& 0\end{array}]^H,\\
\label{CandQ} \mathbf{q}&=[\begin{array}{cc}q_1 &q_2\end{array}]^T,&\mathbf{p}&=[\begin{array}{cc}p_1& p_2\end{array}]^T.
\end{align} 
\hrulefill
\end{figure*}
where $p_k$, $q_k$, $\sigma_{\eta_{jk}}^2$, and $\tilde{\sigma}_{\eta_{jk}}^2$ for $i,j\in\{1,2\}$ are, respectively, the transmission power of user $k$,the  complementary variance of the transmitted signal of user $k$, the aggregated variance and the complementary variance of the HWD noise in  the link between user $j$ and user $k$. 
 The rate of user $k\in\{1,2\}$ can be written using vector notation as 
\setcounter{equation}{3}
\begin{equation}\label{ak-0}
R_k=\frac{1}{2}\log_2\left(\frac{(\sigma^2+\mathbf{a}_k^T\mathbf{p})^2-|\mathbf{f}_k^H\mathbf{q}+\tilde{\mathbf{f}}_k^H\mathbf{p}|^2}{(\sigma^2+\mathbf{b}_k^T\mathbf{p})^2-|\mathbf{g}_k^H\mathbf{q}+\tilde{\mathbf{f}}^H_k\mathbf{p}|^2}\right),
\end{equation}
where 
the corresponding parameters are defined in (5)-(9), shown at the top of the next page.
We also define 
$\Omega\!=\!\{p_k,q_k\!:0\leq p_k\leq P_k, |q_k|\leq p_k,k=1,2\}$ as the feasible set of the design parameters, where $P_k$ is the power budget of user $k$. Note that since $q_k$ for $k=1,2$ is the complementary variance of user $k$, its absolute value has to be not greater than the transmission power of user $k$, i.e., $|q_k|\leq p_k$.

It is to be noted that, 
in practice, discrete rather than Gaussian signaling is employed (see, e.g., \cite{santamaria2018information, javed2019asymmetric, nguyen2015improper}), which will lead to performance degradation  with respect to IGS. 
The significance of studying  improper Gaussian signals is that it shows us whether {\em improper} signaling may in principle achieve performance improvements over {\em proper} signaling.  
In this paper, we focus on IGS and leave the analysis and design of improper discrete constellations for future work. 
\subsection{Problem Statement}
In this paper, we aim at obtaining the boundary of the achievable rate region for the described two-user IC. 
To this end, we employ the following definition of the Pareto boundary for the achievable rate region. 
\begin{definition}[\!\cite{lameiro2017rate, jorswieck2008complete}]
The rate pair ($R_1,R_2$) is called Pareto-optimal if ($R_1^{\prime},R_2$) and ($R_1,R_2^{\prime}$), with $R_1^{\prime}>R_1$ and $R_2^{\prime}>R_2$ , are not achievable.
\end{definition}

The rate region is the union of all these achievable rate tuples, i.e., 
$\mathcal{R}= \underset{\{\mathbf{p},\mathbf{q}\}\in\Omega}{\bigcup}(R_1,R_2)$, 
and its boundary can be derived by the rate profile technique as in the following  optimization problem \cite{zeng2013transmit}
\setcounter{equation}{9}
\begin{subequations}
\begin{align}
 \underset{R,\mathbf{p},\mathbf{q}
 }{\text{maximize}}\,\,\,\,\,\,\,\,  & 
  R &\\
 \label{5b} \text{s.t.}  \,\,\,\,\,\,\,\, \,\,\,\,\,\,\,\,\,\,\,\,\,\,\,\,&  R_{k}\geq\lambda_k R,&k=1,2,
  \\ & 0\leq p_k\leq P_k,&k=1,2,\label{P-Const} \\
 \label{k-Const} &  |q_k|\leq p_k,
 &k=1,2,
 \end{align}
\label{rateregion1}
\end{subequations}
\!where $\lambda_1,\lambda_2\geq 0$ are fixed and $\lambda_1+\lambda_2=1$. 
We can obtain the boundary of the rate region by solving \eqref{rateregion1} for different rate-profile parameters, i.e., $\lambda_1$ and $\lambda_2$.
Note there are efficient algorithms to derive the global optimal solution of convex optimization problems \cite{boyd2004convex,aubry2018new,yang2017unified}. 
However, we are unable to apply these algorithms to \eqref{rateregion1} 
due to the fact that the rates are not concave functions of the optimization variables, which makes \eqref{rateregion1}  non-convex \cite{boyd2004convex,aubry2018new,yang2017unified}. 
Hence, in this paper we propose numerical algorithms to derive suboptimal solutions to \eqref{rateregion1}.

The paper \cite{javed2018improper} proposed an algorithm based on DCP to maximize the achievable rate of a multihop relay system with  additive asymmetric HWD, in which all nodes transmit with maximum power, by optimizing over the complementary variances. 
Such algorithms cannot be applied for a joint optimization of powers and complementary variances since, in this case,  the rates are not a difference of two jointly concave/convex functions in $\mathbf{p}$ and  $\mathbf{q}$. Hence, we solve \eqref{rateregion1} by  MM and FP. 
In MM, the objective function and constraints of an optimization problem are not required to follow a very specific structure such as being a difference of two convex/concave functions, which makes it more powerful than DCP. 

To solve \eqref{rateregion1}, we rewrite it 
such that 
it is more suitable to be solved with MM and FP. 
To this end, we employ the PSINR profile technique in \cite{lameiro2013amplify,santamaria2010capacity} to write an optimization problem that results in the solution of \eqref{rateregion1}.
We define the PSINR profile as 
\begin{subequations}
\begin{align}
 \underset{E,\mathbf{p},\mathbf{q}
 }{\text{maximize}}\,\,\,\,\,\,\,\,  & 
  E &\\
 \label{5b} \text{s.t.}  \,\,\,\,\,\,\,\, \,\,\,\,\,\,\,\,\,\,\,\,\,\,\,\,&  E_k(\mathbf{p},\mathbf{q})\geq 1+\alpha_k E,&k=1,2,
  \\ 
  & 0\leq p_k\leq P_k,&k=1,2, \\
 &  |q_k|\leq p_k,
 &k=1,2,
 \end{align}
\label{enpr}
\end{subequations}
where $\alpha_1\geq 0$ and $\alpha_2\geq 0$ are constants,  $\alpha_1+\alpha_2=1$, and
\begin{equation}\label{q-26-en}
E_k(\mathbf{p},\mathbf{q})\triangleq\frac{(\sigma^2+\mathbf{a}_k^T\mathbf{p})^2-|\mathbf{f}_k^H\mathbf{q}+\tilde{\mathbf{f}}_k^H\mathbf{p}|^2}{(\sigma^2+\mathbf{b}_k^T\mathbf{p})^2-|\mathbf{g}_k^H\mathbf{q}+\tilde{\mathbf{f}}^H_k\mathbf{p}|^2}=\frac{u_k(\mathbf{p},\mathbf{q})}{v_k(\mathbf{p},\mathbf{q})}.
\end{equation}
We can derive the  boundary of the PSINR region by varying $\alpha_1\in [0,1]$. 
Note that $E_k(\mathbf{p},\mathbf{q})\geq 1$ for $k=1,2$ since the rates are non-negative. 
Moreover, the numerator and denominator of $E_k(\mathbf{p},\mathbf{q})$ are strictly positive because the rates are bounded and non-negative. 
In the following lemma, we show that this technique results in the boundary of the rate region in \eqref{rateregion1}.
\begin{lemma}[\!\cite{lameiro2013amplify,santamaria2010capacity}]
Every point in the boundary of the
rate region corresponds to a point in the boundary of the PSINR region, and vice versa.
\end{lemma}
\begin{proof}
Assume there exists a pair $(R_1,R_2)$ on the boundary of the achievable rate region that is not on the boundary of the PSINR region. 
In other words, the pair $(E_1=2^{R_1},E_2=2^{R_2})$, which is a feasible PSINR pair, is not on the boundary of the PSINR region, and hence there exist $E_1^{\prime}$ and/or $E_2^{\prime}$ such that the pairs $(E_1^{\prime}>E_1,E_2)$ and/or $(E_1,E_2^{\prime}>E_2)$ are feasible.
Since the logarithm functions are monotonically increasing, the rate pairs $(0.5\log_2(E_1^{\prime})>R_1,R_2)$ or $(R_1,0.5\log_2(E_2^{\prime})>R_2)$ are achievable, which implies that $(R_1,R_2)$ is not on the boundary of the rate region.  Similarly, it can be shown that every point in the boundary of the PSINR region associates with a point in the boundary of the
rate region.
\end{proof}
Note that 
we can rewrite \eqref{enpr} as the following maximin optimization problem by removing the variable $E$
\begin{equation}
 \underset{0\leq p_k\leq P_k,|q_k|\leq p_k
 }{\text{maximize}}\,\,\,\,\,\,\,\,   
  \underset{k=1,2}{\min}\left\{\frac{E_k(\mathbf{p},\mathbf{q})-1}{\alpha_k}\right\}. 
\label{enpr-mxmn}
\end{equation}

\section{Boundary of the rate region by Fractional Programming}\label{FP}

In this section, we solve the PSINR profile problem in \eqref{enpr} by MM, which results in solving a sequence of fractional optimization problems. 
We solve each fractional optimization problem by FP and the generalized Dinkelbach algorithm \cite{dinkelbach1967nonlinear,zappone2015energy,shen2018fractional}.
Our proposed algorithm converges to a stationary point of \eqref{enpr}. 
We first provide preliminaries on generalized Dinkelbach's algorithm in Section \ref{dinkel-alg} and then propose our algorithm in \ref{sec-iii-b}.

\subsection{Preliminaries of generalized Dinkelbach's algorithm}\label{dinkel-alg}
 Dinkelbach's algorithm is a powerful tool that solves FP problems, which was proposed to handle single-ratio functions. The generalized Dinkelbach algorithm is a modified Dinkelbach algorithm to solve maximin multiple-ratio problems \cite{crouzeix1991algorithms}. 
  The generalized Dinkelbach algorithm is an iterative approach, in which the fractional functions are approximated by surrogate functions at each iteration. 
  In the following lemma, we present some conditions that are used in the generalized Dinkelbach algorithm. 
  \begin{lemma}[\!\!\cite{crouzeix1991algorithms,zappone2015energy}]
Consider the fractional functions $\frac{u_i(\mathbf{x})}{v_i(\mathbf{x})}$, where $u_i(\mathbf{x})$ and $v_i(\mathbf{x})$ are continuous in $\mathbf{x}$, $v_i(\mathbf{x})$ is strictly positive in $\mathbf{x}$, and $\mathbf{x}$ is a vector with dimension $n$ that belongs to a compact set $\mathcal{X}$. 
Let us define 
\begin{align}\label{surr-dink}
V(\mu)&=\underset{\mathbf{x}}{\max}\,\,\,\underset{i}{\min}\left(u_i(\mathbf{x})-\mu v_i(\mathbf{x})\right),\\
\label{surr-dink-2}\bar{\mu}&=\underset{\mathbf{x}}{\max}\,\,\,\underset{i}{\min}\left(\frac{u_i(\mathbf{x})}{v_i(\mathbf{x})}\right),
\end{align}
where $V(\mu)$, $\bar{\mu}$, and $\mu$ 
are real and scalar, and have the following properties. 
\begin{enumerate}
  \item $V(\mu)$ is continuous and strictly decreasing in $\mu$.
  \item The optimization problems \eqref{surr-dink} and \eqref{surr-dink-2} always have optimal solutions.
  \item $\bar{\mu}$ is finite and $V(\bar{\mu})=0$.
  \item $V(\mu)$ has a unique root, and $V(\mu)=0$ implies $\mu=\bar{\mu}$.
\end{enumerate}
\end{lemma}
The generalized Dinkelbach algorithm 
employs the surrogate function $V(\mu)$ in \eqref{surr-dink} and tries to iteratively find the unique root of $V(\mu)$, i.e., $\bar{\mu}$.
The algorithm starts with an initial point, e.g., $\mu^{(0)}=0$, then it updates $\mu$ to obtain $\bar{\mu}$.
Assuming $u_i(\mathbf{x})\geq 0$, which is the case we consider in this paper, $V(0)=\underset{\mathbf{x}}{\max}\,\,\,\underset{i}{\min}\left(u_i(\mathbf{x})\right)>0$.
Since $V(\mu)$ is continuous and strictly decreasing in $\mu$, $\mu$ is chosen monotonically increasing at each iteration ($\mu^{(l)}>\mu^{(l-1)}$) 
 until $V(\mu)$ approaches 0. 
At the $l$th iteration, $\mu^{(l)}$ is 
\begin{equation}\label{mu-eq}
\mu^{(l)}=\min\left(\frac{u_1(\mathbf{x}^{(l-1)})}{v_1(\mathbf{x}^{(l-1)})},\frac{u_2(\mathbf{x}^{(l-1)})}{v_2(\mathbf{x}^{(l-1)})}\right)>0,\end{equation} 
where $\mathbf{x}^{(l-1)}$ is 
 \begin{align}\label{su-dunk}
\mathbf{x}^{(l-1)}=\underset{\mathbf{x}}{\arg\max}\,\,\,\underset{i}{\min}\left(u_i(\mathbf{x})-\mu^{(l-1)} v_i(\mathbf{x})\right).
\end{align}
The generalized Dinkelbach algorithm updates $\mu^{(l)}$ and $\mathbf{x}^{(l-1)}$ based on \eqref{mu-eq} and \eqref{su-dunk}, respectively, until a convergence metric is met, e.g., $V(\mu^{(l)})<\epsilon$, where $\epsilon>0$. 
This algorithm converges linearly to the optimal solution \cite{crouzeix1991algorithms}. 

 Note that in order to apply the generalized Dinkelbach algorithm, 
 it is not required that $u_i(\mathbf{x})$ and $v_i(\mathbf{x})$ fulfill any other condition (except those in the lemma), 
 which makes this algorithm a powerful tool to solve different types of fractional problems. 
 If $u_i(\mathbf{x})$ and $v_i(\mathbf{x})$ are concave and convex functions, respectively,  the optimization problem at each iteration is convex and can  easily be solved. 
 However, in the general case, it might be difficult to efficiently solve the optimization problem at each iteration. 

\subsection{Proposed algorithm}\label{sec-iii-b}
We can apply the generalized Dinkelbach algorithm to derive the boundary of the PSINR region since the optimization problem can be written as a maximin weighted problem as indicated in \eqref{enpr-mxmn}.  
However, since $u_k(\mathbf{p},\mathbf{q})$ and $v_k(\mathbf{p},\mathbf{q})$ are not, respectively, concave and convex in optimization variables, the corresponding optimization problem in each iteration of the generalized Dinkelbach algorithm is not convex. 
Indeed, $u_k(\mathbf{p},\mathbf{q})$ and $v_k(\mathbf{p},\mathbf{q})$ are a difference of two convex/concave functions:
\begin{align}\label{ufunk}
u_k(\mathbf{p},\mathbf{q})&
=\underbrace{-|\mathbf{f}_k^H\mathbf{q}+\tilde{\mathbf{f}}_k^H\mathbf{p}|^2}_{\text{concave part}}+\underbrace{(\sigma^2+\mathbf{a}_k^T\mathbf{p})^2}_{\text{convex part}},\\
v_k(\mathbf{p},\mathbf{q})&
=\underbrace{-|\mathbf{g}_k^H\mathbf{q}+\tilde{\mathbf{f}}^H_k\mathbf{p}|^2}_{\text{concave part}}+\underbrace{ (\sigma^2+\mathbf{b}_k^T\mathbf{p})^2}_{\text{convex part}}.
\label{vfunk}
\end{align}
Hence, to solve \eqref{enpr-mxmn}, we employ a sequential optimization approach  by approximating $E_k(\mathbf{p},\mathbf{q})$ with a lower bound  $\tilde{E}_k(\mathbf{p},\mathbf{q},\mu)$ in each iteration  \cite{sun2017majorization, aubry2018new}. 
Then, we obtain the global optimal solution of each surrogate optimization problem by the generalized Dinkelbach algorithm. 
To this end, in each iteration, we first approximate  $u_k(\mathbf{p},\mathbf{q})$ by a lower bound concave function $\tilde{u}_k(\mathbf{p},\mathbf{q})$ and $v_k(\mathbf{p},\mathbf{q})$ by an upper bound convex function $\tilde{v}_k(\mathbf{p},\mathbf{q})$ as in the following lemma. 
\begin{lemma}\label{uvfunk-app}
A concave lower bound for $u_k(\mathbf{p},\mathbf{q})$ in the $m$th iteration is 
\begin{align}\nonumber
\tilde{u}_k^{(m)}(\mathbf{p},\mathbf{q})=&-|\mathbf{f}_k^H\mathbf{q}+\tilde{\mathbf{f}}_k^H\mathbf{p}|^2
+(\sigma^2+\mathbf{a}_k^T\mathbf{p}^{(m)})^2\\&
+2(\sigma^2+\mathbf{a}_k^T\mathbf{p}^{(m)})\mathbf{a}_k^T(\mathbf{p}-\mathbf{p}^{(m)}),
 \end{align}
Moreover, a convex upper bound for $v_k(\mathbf{p},\mathbf{q})$ in the $m$th iteration is
 \begin{align}
\nonumber \tilde{v}_k^{(m)}(\mathbf{p},\mathbf{q})=&(\sigma^2+\mathbf{b}_k^T\mathbf{p})^2-|\mathbf{g}_k^H\mathbf{q}^{(m)}+\tilde{\mathbf{f}}^H_k\mathbf{p}^{(m)}|^2\\&\nonumber
-2\mathfrak{R}\left[\tilde{\mathbf{f}}_k^H(\mathbf{g}_k^H\mathbf{q}^{(m)} +\tilde{\mathbf{f}}^H_k\mathbf{p}^{(m)})^*\right](\mathbf{p}-\mathbf{p}^{(m)})\\
&-2\mathfrak{R}\left[(\mathbf{g}_k^H\mathbf{q}^{(m)}+\tilde{\mathbf{f}}^H_k\mathbf{p}^{(m)})^*\mathbf{g}_k^H(\mathbf{q}-\mathbf{q}^{(m)})\right],
 \end{align}
where $\mathbf{p}^{(m)}$ and $\mathbf{q}^{(m)}$ are the power and complementary variances at the $m$th iteration, which are the solution of the previous iteration. Furthermore, $\mathfrak{R}\left[x\right]$ takes the real part of $x$.
\end{lemma}
\begin{proof}
Please refer to Appendix \ref{app-new}. 
\end{proof}

Now, we are able to write the surrogate optimization problem in $m$th iteration as 
\begin{subequations}
\begin{align}
 \underset{E^{\prime},\mathbf{p},\mathbf{q}
 }{\text{maximize}}\,\,\,\,\,\,\,\,  & 
  E &\\
 \label{29b-2} \text{s.t.}  \,\,\,\,\,\,\,\, \,\,\,\,\,\,\,\,\,\,\,\,\,\,\,\,&  \tilde{E}_k^{(m)}(\mathbf{p},\mathbf{q})\geq 1+\alpha_kE^{\prime},&k=1,2,
  \\ \label{P-Const-29}& 0\leq p_k\leq P_k,&k=1,2, \\
 \label{k-Const-29} &  |q_k|\leq p_k,
 &k=1,2,
 \end{align}
\label{sikim-2}
\end{subequations}
\!where $\tilde{E}_k^{(m)}(\mathbf{p},\mathbf{q})=\frac{\tilde{u}_k^{(m)}(\mathbf{p},\mathbf{q})}{\tilde{v}_k^{(m)}(\mathbf{p},\mathbf{q})}$ 
 and $E_k(\mathbf{p},\mathbf{q})$ 
 fulfill the following conditions:
\begin{enumerate}
\item $\tilde{E}_k^{(m)}(\mathbf{p},\mathbf{q})\leq E_k(\mathbf{p},\mathbf{q})$ for all feasible $\mathbf{p},\mathbf{q}$ and $k=1,2$.
\item $\tilde{E}_k^{(m)}(\mathbf{p}^{(m)},\mathbf{q}^{(m)})= E_k(\mathbf{p}^{(m)},\mathbf{q}^{(m)})$ for $k=1,2$.
\item $\frac{\partial \tilde{E}_k^{(m)}(\mathbf{p}^{(m)},\mathbf{q}^{(m)})}{\partial\mathbf{p}}=\frac{\partial E_k(\mathbf{p}^{(m)},\mathbf{q}^{(m)})}{\partial\mathbf{p}}$ and $\frac{\partial \tilde{E}^{(m)}_k(\mathbf{p}^{(m)},\mathbf{q}^{(m)})}{\partial\mathbf{q}}=\frac{\partial E_k(\mathbf{p}^{(m)},\mathbf{q}^{(m)})}{\partial\mathbf{q}}$ for $k=1,2$.
\end{enumerate}
These properties guarantee that the algorithm converges to a stationary point of \eqref{enpr} \cite[Section II.B]{aubry2018new}.  To solve \eqref{sikim-2} and obtain $\mathbf{p}^{(m+1)}$ and $\mathbf{q}^{(m+1)}$, we employ the generalized Dinkelbach algorithm, which gives the global optimal solution of \eqref{sikim-2}, as in the following. 
We summarize this procedure in Algorithm I.

\begin{table}
\label{alg-1}
\begin{tabular}{|l|}
\hline \small{\textbf{Algorithm I} Proposed sequential optimization algorithm.}\\
\hline\small{\textbf{Initialization}}\\
\small{Set $\epsilon$, $M$, $\mathbf{p}^{(0)}=\mathbf{0}$, $\mathbf{q}^{(0)}=\mathbf{0}$, $m=1$, convergence=0}\\
\hline \small{\textbf{While} convergence=0 and $m\leq M$ \textbf{do}}\\
\,\,\,\,\,\,\,\small{Construct $\tilde{E}_k^{(m)}(\mathbf{p},\mathbf{q})=\tilde{u}_k^{(m)}(\mathbf{p},\mathbf{q})/\tilde{v}_k^{(m)}(\mathbf{p},\mathbf{q})$ for $k=1,2$ }\\
\,\,\,\,\,\,\,\,\,\,\,\,\,\, \small{using Lemma \ref{uvfunk-app}}\\
\,\,\,\,\,\,\,\small{Obtain $\mathbf{p}^{(m+1)}$ and $\mathbf{q}^{(m+1)}$ by solving \eqref{sikim-2}, i.e.,}\\
\,\,\,\,\,\,\,\,\,\,\,\,\,\, \small{run algorithm II}\\
\,\,\,\,\,\,\,\small{\textbf{If} $\| \mathbf{p}^{(m)}-\mathbf{p}^{(m+1)}\| /\| \mathbf{p}^{(m)}\| <\epsilon$} \\
\,\,\,\,\,\,\,\,\,\,\small{and $\| \mathbf{q}^{(m)}-\mathbf{q}^{(m+1)}\| / \|\mathbf{q}^{(m)}\| <\epsilon$}\\
\,\,\,\,\,\,\,\,\,\,\,\,\,\,\,\,\,\small{convergence=1}\\
\,\,\,\,\,\,\,\,\,\,\,\,\,\,\,\,\,\small{$\mathbf{p}^{\star}=\mathbf{p}^{(m+1)}$ and $\mathbf{q}^{\star}=\mathbf{q}^{(m+1)}$}\\
\,\,\,\,\,\,\,\small{\textbf{End (If)}}\\
\,\,\,\,\,\,\,\small{$m=m+1$}\\
\small{\textbf{End (While)}}\\
\small{{\bf Return} $\mathbf{p}^{\star}$ and $\mathbf{q}^{\star}$.}\\
\hline
\end{tabular}
\end{table}

 Now we solve \eqref{sikim-2} and obtain its global optimal solution by FP, which is also an iterative algorithm as explained in Section \ref{dinkel-alg}. 
 To this end, we introduce the following functions, which are the corresponding surrogate functions of $\frac{\tilde{E}^{(m)}_k-1}{\alpha_k}$ for $k=1,2$:
 \begin{equation}\label{E-tilde}
\hat{E}_k(\mathbf{p},\mathbf{q},\mu^{(l)})\triangleq u_k^{(m)}(\mathbf{p},\mathbf{q})-(\mu^{(l)}\alpha_k+1)v_k^{(m)}(\mathbf{p},\mathbf{q}),
\end{equation}
where $\mu^{(l)} \in\mathbb{R}$ is fixed and given by
 \begin{equation}\label{mu-31}
\mu^{(l)}=\min_{k=1,2}\left(\frac{\tilde{E}_k(\mathbf{p}^{(l-1)},\mathbf{q}^{(l-1)})-1}{\alpha_k}\right).
\end{equation}
It is worth mentioning that the generalized Dinkelbach algorithm requires an initial point $\mu^{(0)}$, which can be obtained by substituting $\mathbf{p}^{(m)}$ and $\mathbf{q}^{(m)}$ in \eqref{mu-31}. 
By substituting  \eqref{E-tilde} in \eqref{sikim-2}, the optimization problem at each iteration of the generalized Dinkelbach algorithm is
\begin{subequations}
\begin{align}
 \underset{E^{\prime},\mathbf{p},\mathbf{q}
 }{\text{maximize}}\,\,\,\,\,\,\,\,  & 
  E^{\prime} &\\
 \label{29b} \text{s.t.}  \,\,\,\,\,\,\,\, \,\,\,\,\,\,\,\,\,\,\,\,\,\,\,\,&  \hat{E}_k(\mathbf{p},\mathbf{q},\mu^{(l)})\geq E^{\prime},&k=1,2,
  \\ & \eqref{P-Const-29},\eqref{k-Const-29}.
 \end{align}
\label{enpr-eq}
\end{subequations}
 We solve \eqref{enpr-eq} for the given $\mu^{(l)}$, which results in $\mathbf{p}^{(l)}$ and $\mathbf{q}^{(l)}$.
 Then, we update $\mu^{(l)}$ by \eqref{mu-31} and repeat the procedure until a convergence metric is met. 
 As indicated in Section \ref{dinkel-alg}, the convergence rate of the generalized Dinkelbach algorithm is linear.
 The optimization problem \eqref{enpr-eq} is convex, and its global optimal solution can be efficiently obtained  \cite{boyd2004convex}. 
 We summarize this procedure in Algorithm II.
 \begin{table}
\label{alg-2}
\begin{tabular}{|l|}
\hline \small{\textbf{Algorithm II} Generalized Dinkelbach algorithm.}\\
\hline\small{\textbf{Initialization}}\\
\small{Set $\epsilon$, $L$,  $l=0$, $\mu^{(l)}=\min_{k=1,2}\left(\frac{\tilde{E}_k(\mathbf{p}^{(m)},\mathbf{q}^{(m)})-1}{\alpha_k}\right)$}\\
\small{Compute $\hat{E}_k(\mathbf{p},\mathbf{q},\mu^{(l)})$ for $k=1,2$ by \eqref{E-tilde}}\\
\hline \small{\textbf{While} $\underset{k=1,2}{\min}\{\hat{E}_k(\mathbf{p},\mathbf{q},\mu^{(l)})\}\geq \epsilon$ and $l\leq L$ \textbf{do}}\\
\,\,\,\,\,\,\,\small{$l=l+1$}\\
\,\,\,\,\,\,\,\small{Obtain $\mathbf{p}^{(l)}$ and $\mathbf{q}^{(l)}$ by solving \eqref{enpr-eq}}\\
\,\,\,\,\,\,\,\small{\textbf{If} $\underset{k=1,2}{\min}\{\tilde{E}_k(\mathbf{p},\mathbf{q},\mu^{(l)})\}<\epsilon$}\\
\,\,\,\,\,\,\,\,\,\,\,\,\,\,\,\,\,\small{$\mathbf{p}^{\star}=\mathbf{p}^{(l)}$ and $\mathbf{q}^{\star}=\mathbf{q}^{(l)}$}\\
\,\,\,\,\,\,\,\small{\textbf{Else}}\\
\,\,\,\,\,\,\,\,\,\,\,\,\,\,\,\,\,\small{Update $\mu^{(l)}$ by \eqref{mu-31}}\\
\,\,\,\,\,\,\,\small{\textbf{End (If)}}\\
\small{\textbf{End (While)}}\\
\small{{\bf Return} $\mathbf{p}^{\star}$ and $\mathbf{q}^{\star}$.}\\
\hline
\end{tabular}
\end{table}

To sum up, the proposed algorithm works as follows. 
We solve the PSINR profile in \eqref{enpr} by solving a sequence of fractional optimization problems.  
Indeed, we employ a sequential optimization approach and approximate the PSINR term of each user by a lower bound. 
In order to derive the global optimal solution of each  fractional optimization problem, 
 we perform another iterative algorithm, i.e., the generalized Dinkelbach algorithm. 
It is worth mentioning that this algorithm does not converge to the Pareto-optimal solution; however, it obtains a stationary point of \eqref{enpr}.

\section{Simplified algorithm}\label{sub-Sec}
In this section, we propose a simplified version of the algorithm from Section \ref{FP}, which exhibits a lower computational complexity. 
In the simplified algorithm, we first optimize the transmission power $\mathbf{p}$ for PGS, i.e., for $\mathbf{q}=\mathbf{0}$. 
This problem is addressed in Section \ref{Sec-PGS}. 
Then, in Section \ref{sub-Sec}.B, we optimize the complementary variances for the resulting transmit power $\mathbf{p}$ such that the rates of all users is simultaneously increased.

\subsection{Power optimization}\label{Sec-PGS}
 In this subsection, we optimize the transmission power vector $\mathbf{p}$ for PGS, i.e., when $\mathbf{q}=\mathbf{0}$. In this case, deriving the boundary of the PSINR region can be cast as the optimization problem
 \begin{subequations}
 \begin{align}
 \underset{E,\mathbf{p}
 }{\text{maximize}}\,\,\,  & 
  E &\\
 \label{39b} \text{s.t.}   \,\,\,\,\,\,\,\,\,\,\,\,\,\,\,\,&  \frac{(\sigma^2+\mathbf{a}_i^T\mathbf{p})^2-|\tilde{\mathbf{f}}_i^H\mathbf{p}|^2}{(\sigma^2+\mathbf{b}_i^T\mathbf{p})^2-|\tilde{\mathbf{f}}^H_i\mathbf{p}|^2}\geq 1+\alpha_k E,
 &\!\!k=1,2,
  \\ \label{P-Const-39}& 0\leq p_k\leq P_k,&\!\!k=1,2,
 \end{align}
\label{enpr-PGS}
\end{subequations}
 for $\alpha_1,\alpha_2\geq 0$ and $\alpha_1+\alpha_2=1$. Unfortunately, the optimization problem in \eqref{enpr-PGS} is not convex due to \eqref{39b}. 
 In the following lemma, we derive a lower bound for \eqref{39b}, which allows us to simplify \eqref{enpr-PGS} and derive a low-complexity algorithm. 
\begin{lemma}\label{E-lower-1}
A lower bound for the left-hand side of  \eqref{39b} is
\begin{equation}\label{lo-en-bo}
 \frac{(\sigma^2+\mathbf{a}_i^T\mathbf{p})^2-|\tilde{\mathbf{f}}_i^H\mathbf{p}|^2}{(\sigma^2+\mathbf{b}_i^T\mathbf{p})^2-|\tilde{\mathbf{f}}^H_i\mathbf{p}|^2}\geq \frac{(\sigma^2+\mathbf{a}_i^T\mathbf{p})^2}{(\sigma^2+\mathbf{b}_i^T\mathbf{p})^2},
\end{equation}
where the equality in \eqref{lo-en-bo} holds if and only if the HWD noise is proper, i.e., $\tilde{\mathbf{f}}_i=\mathbf{0}$.
\end{lemma}
\begin{proof}
It is easy to verify that $0\leq|\tilde{\mathbf{f}}_i^H\mathbf{p}|^2<(\sigma^2+\mathbf{b}_i^T\mathbf{p})^2<(\sigma^2+\mathbf{a}_i^T\mathbf{p})^2$.
Let us define 
\begin{equation}
f(t)=\frac{\beta_1-t}{\beta_2-t},
\end{equation}
where $0\leq t<\beta_2<\beta_1$. 
The lower bound in \eqref{lo-en-bo} is then satisfied if $f(t)$ is increasing in $t$. This function is strictly increasing in $t\in [0,\beta_2)$ since
\begin{equation}
\frac{\partial f(t)}{\partial t}=\frac{\beta_1-\beta_2}{(\beta_2-t)^2}>0,
\end{equation}
Thus, we have 
\begin{equation}
\frac{\beta_1-t}{\beta_2-t}\geq\frac{\beta_1}{\beta_2},
\end{equation}
with equality if and only if $t=0$. 
\end{proof}
For each point characterized by $\alpha_1$ and $\alpha_2$, we solve \eqref{enpr-PGS} for the lower bound in \eqref{lo-en-bo} as the optimization problem
\begin{subequations}
 \begin{align}
 \underset{E,\mathbf{p}
 }{\text{maximize}}\,\,\,\,\,\,\,\,  & 
  E &\\
 \label{46b} \text{s.t.}  \,\,\,\,\,\,\,\, \,\,\,\,\,\,\,\,\,\,\,\,\,\,\,\,&  \frac{\sigma^2+\mathbf{a}_k^T\mathbf{p}}{\sigma^2+\mathbf{b}_k^T\mathbf{p}}\geq \sqrt{1+\alpha_k E},
 &k=1,2,
  \\ \label{P-Const-39}& 0\leq p_k\leq P_k,&k=1,2.
 \end{align}
\label{enpr-PGS-app}
\end{subequations}
It is worth mentioning that the lower bound in Lemma \ref{E-lower-1} is employed to simplify \eqref{enpr-PGS} and obtain the powers, and the actual rates are derived by substituting the obtained powers in \eqref{R1}. 
Note that the region achieved by solving  \eqref{enpr-PGS} includes the region achieved by solving \eqref{enpr-PGS-app}.
If the additive HWD noise is proper, \eqref{enpr-PGS-app} is equivalent to \eqref{enpr-PGS}\footnote{ 
This is in line with \cite{hellings2017worst}, where it was shown that proper Gaussian noise is the worst case in a $K$-user MIMO IC with ideal devices.}. 
The global optimum solution of \eqref{enpr-PGS-app} can be derived by employing a bisection method and solving a sequence of feasibility problems \cite{aubry2018new}. 
That is, we fix $E$ as $E^{\prime}$ and 
consider the feasibility problem \eqref{enpr-PGS-app-bi}, shown at the top of the next page.
\begin{figure*}
\begin{subequations}
\begin{align}
 \text{find}\,\,\,\,\,\,\,\, \,\,\,&\mathbf{p}\in\mathbb{R}^{2},\\
   \text{s.t.}  \,\,\,\,\,\,\,\,\,\,\,\,\,&  (\mathbf{a}_k^T-\sqrt{1+\alpha_k E^{\prime}}\mathbf{b}_k^T)\mathbf{p}\geq (\sqrt{1+\alpha_k E^{\prime}}-1) \sigma^2,&k=1,2,\label{e-Const-40}
  \\ \label{P-Const-40}& 0\leq p_k\leq P_k,&k=1,2.
 \end{align}
\label{enpr-PGS-app-bi}
\end{subequations}
\setcounter{equation}{33}
\begin{align}
\label{line-slop}
\mathbf{A}&=\left[ \begin{array}{c} \mathbf{a}_1^T-\sqrt{1+\alpha_1 E^{\prime}} \mathbf{b}_1^T \\ \mathbf{a}_2^T-\sqrt{1+\alpha_2 E^{\prime}}\mathbf{b}_2^T \end{array} \right]=\left[ \begin{array}{cc} 
|h_{11}|^2\left(1-\sigma^2_{\eta_{11}}(\sqrt{1+\alpha_1 E^{\prime}}-1)\right)&-|h_{21}|^2(1+\sigma^2_{\eta_{21}})(\sqrt{1+\alpha_1 E^{\prime}}-1)
\\ 
-|h_{12}|^2(1+\sigma^2_{\eta_{12}})(\sqrt{1+\alpha_2 E^{\prime}}-1)&|h_{22}|^2\left(1-\sigma^2_{\eta_{22}}(\sqrt{1+\alpha_2 E^{\prime}}-1)\right)
 \end{array} \right].
\end{align}
\setcounter{equation}{34}
\begin{subequations}
 \begin{align}
 \underset{t,\mathbf{q}}{\text{maximize}}\,\,\,\,\,\,\,\,  & 
  t &\\
 \label{49b} \text{s.t.}  \,\,\,\,\,\,\,\, \,\,\,\,\,\,\,\,\,\,\,\,\,\,\,\,&  \frac{(\sigma^2+\mathbf{a}_k^T\mathbf{p}^{\star})^2-|\mathbf{f}_k^H\mathbf{q}+\tilde{\mathbf{f}}_k^H\mathbf{p}^{\star}|^2}{(\sigma^2+\mathbf{b}_k^T\mathbf{p}^{\star})^2-|\mathbf{g}_k^H\mathbf{q}+\tilde{\mathbf{f}}^H_k\mathbf{p}^{\star}|^2}\geq E_{p,k}+\alpha_kt,&k=1,2,
  \\ \label{ka-Const-49}& |q_k|\leq p_k^{\star},
  &k=1,2.
 \end{align}
\label{enpr-iGS}
\end{subequations}
\setcounter{equation}{37}
\begin{subequations}
 \begin{align}
 \underset{t_1,t_2,\mathbf{q}}{\text{maximize}}\,\,\,\,\,\,\,\,  & 
  \min (t_1,t_2) &\\
 \label{55b} \text{s.t.}  \,\,\,\,\,\,\,\, \,\,\,\,\,\,\,\,\,\,\,\,\,\,\,\,&  \frac{(\sigma^2+\mathbf{a}_k^T\mathbf{p}^{\star})^2-|\mathbf{f}_k^H\mathbf{q}+\tilde{\mathbf{f}}_k^H\mathbf{p}^{\star}|^2-\alpha_k t_k}{(\sigma^2+\mathbf{b}_k^T\mathbf{p}^{\star})^2-|\mathbf{g}_k^H\mathbf{q}+\tilde{\mathbf{f}}^H_k\mathbf{p}^{\star}|^2}\geq E_{p,k},&k=1,2,
  \\ \label{ka-Const-59}& |q_k|\leq p_k^{\star},
  &k=1,2.
 \end{align}
\label{enpr-iGS-2}
\end{subequations}
\hrulefill
\end{figure*}
If \eqref{enpr-PGS-app-bi} is feasible for a given $E^{\prime}$, the optimal solution of \eqref{enpr-PGS-app} is greater than or equal to $E^{\prime}$, i.e., $E^{\star}\geq E^{\prime}$. Otherwise, $E^{\star}< E^{\prime}$.
In order to find $E^{\star}$, we employ the well-known bisection method over $E^{\prime}$ solving \eqref{enpr-PGS-app-bi} at each iteration, which yields, upon convergence, the global optimal solution of \eqref{enpr-PGS-app} \cite{boyd2004convex}. 
Constraints \eqref{e-Const-40} and \eqref{P-Const-40} are linear in $\mathbf{p}$, which permits deriving a closed-form expression for a feasible point, as presented in the following theorem. 
It is worth mentioning that this algorithm does not attain the global optimal solution of \eqref{enpr-PGS}. 
There might be optimization approaches to obtain its global optimal solution such as the monotonic optimization framework \cite{rubinov2001alogrithm, liu2012achieving, zappone2017globally}, although the computational complexity of these approaches is high.
\begin{theorem}\label{E-lower}
The optimization problem in \eqref{enpr-PGS-app-bi} is feasible for a given $E^{\prime}$ if and only if 
$0\leq p_k^{\prime}\leq P_k$, for $k=1,2$, where 
\setcounter{equation}{32}
\begin{align}
\label{int-sec-pnt}\left[ \begin{array}{c} p_1^{\prime} \\ p_2^{\prime} \end{array} \right] &=\mathbf{A}^{-1}\left[ \begin{array}{c}
(\sqrt{1+\alpha_1 E^{\prime}}-1)\sigma^2 \\ (\sqrt{1+\alpha_2 E^{\prime}}-1)\sigma^2 \end{array} \right].
\end{align}
Moreover, $\mathbf{A}$ is given by \eqref{line-slop}, shown at the top of the next page.
\end{theorem}
\begin{proof}
Please refer to Appendix \ref{app:th-1}.
\end{proof}
We note that this algorithm leads to the optimal PGS only when HWD noise is proper. 
Note that PGS is suboptimal, in {\em point-to-point} communications, in the presence of asymmetric HWD \cite{javed2017asymmetric,javed2017impact}. 
Thus, the users may improve the performance by employing IGS in additive asymmetric HWD. 
It is worth noting that, in this paper, we aim at proposing PGS and IGS schemes for the two-user IC with additive asymmetric HWD, but we do not derive sufficient and necessary conditions for the optimality of IGS or PGS in the two-user IC with additive asymmetric HWD, which remains an open problem.

\subsection{Complementary variance design}\label{Sec-cvd}

In this subsection, we optimize the complementary variances $\mathbf{q}$ for a given $\mathbf{p}^{\star}$, which has been obtained by solving \eqref{enpr-PGS-app}. 
We obtain $\mathbf{q}$ such that the rates of both users exceed the rates achieved by PGS, which are the rates achievable with $\mathbf{q}=\mathbf{0}$ and the power vector $\mathbf{p}^\star$ obtained by solving \eqref{enpr-PGS-app}.
 In other words, we want to solve the optimization problem \eqref{enpr-iGS}, shown at the top of this page,
where $p_k^{\star}$ is the $k$th element of $\mathbf{p}^{\star}$. Moreover, $E_{p,k}$ is fixed and given by
\setcounter{equation}{35}
\begin{equation}
E_{p,k}=\frac{(\sigma^2+\mathbf{a}_k^T\mathbf{p}^{\star})^2-|\tilde{\mathbf{f}}_k^H\mathbf{p}^{\star}|^2}{(\sigma^2+\mathbf{b}_k^T\mathbf{p}^{\star})^2-|\tilde{\mathbf{f}}^H_k\mathbf{p}^{\star}|^2}.
\end{equation}
Unfortunately, \eqref{enpr-iGS} is not convex due to \eqref{49b}. Hence, in order to efficiently solve \eqref{enpr-iGS}, we first rewrite \eqref{49b} as 
\begin{equation}
\frac{(\sigma^2+\mathbf{a}_k^T\mathbf{p}^{\star})^2-|\mathbf{f}_k^H\mathbf{q}+\tilde{\mathbf{f}}_k^H\mathbf{p}^{\star}|^2-\alpha_k t_k}{(\sigma^2+\mathbf{b}_k^T\mathbf{p}^{\star})^2-|\mathbf{g}_k^H\mathbf{q}+\tilde{\mathbf{f}}^H_k\mathbf{p}^{\star}|^2}\geq E_{p,k},
\end{equation}
where $t_k=t\left[(\sigma^2+\mathbf{b}_k^T\mathbf{p}^{\star})^2-|\mathbf{g}_k^H\mathbf{q}+\tilde{\mathbf{f}}^H_k\mathbf{p}^{\star}|^2\right]$.
We then relax the relation between $t_1$, $t_2$, and $\mathbf{q}$ and treat $t_1$ and $t_2$ as new optimization variables. In other words, we approximate \eqref{enpr-iGS} as \eqref{enpr-iGS-2}, shown at the top of this page.
If $\min (t_1,t_2)>0$, 
the rates of both users are simultaneously increased by employing IGS. 
Otherwise, 
we set $\mathbf{q}=\mathbf{0}$ and employ PGS. 
Note that the constraint \eqref{55b} can be rewritten as
\setcounter{equation}{38}
\begin{multline}
E_{p,k}|\mathbf{g}_k^H\mathbf{q}+\tilde{\mathbf{f}}^H_k\mathbf{p}^{\star}|^2-|\mathbf{f}_k^H\mathbf{q}+\tilde{\mathbf{f}}_k^H\mathbf{p}^{\star}|^2\\+(\sigma^2+\mathbf{a}_k^T\mathbf{p}^{\star})^2-E_{p,k}(\sigma^2+\mathbf{b}_k^T\mathbf{p}^{\star})^2\geq \alpha_kt_k,\label{59b-app}
\end{multline}
which is a difference of two convex functions.
Thus, \eqref{enpr-iGS-2} is not a convex optimization problem, but it can be efficiently solved by  difference of convex programming and a convex-concave procedure similar to \eqref{enpr-eq}  \cite{duchi2018sequential,lipp2016variations,yuille2003concave,lanckriet2009convergence,sun2017majorization}.  Hence, we employ difference of convex programming (DCP) and solve \eqref{enpr-iGS-2} iteratively. At each iteration, we approximate the left-hand side of \eqref{59b-app} by a concave function. 
To this end, we employ the first-order Taylor expansion and approximate the convex part of \eqref{59b-app} around the point $\mathbf{q}^{(l)}$ by an affine function as
\begin{multline}
|\mathbf{g}_k^H\mathbf{q}+\tilde{\mathbf{f}}^H_k\mathbf{p}^{\star}|^2\simeq |\mathbf{g}_k^H\mathbf{q}^{(l)}+\tilde{\mathbf{f}}^H_k\mathbf{p}^{\star}|^2 \\
+2\mathfrak{R}\left((\mathbf{g}_k^H\mathbf{q}^{(l)}+\tilde{\mathbf{f}}^H_k\mathbf{p}^{\star})^*\mathbf{g}_k^H(\mathbf{q}-\mathbf{q}^{(l)})\right),\label{app-sub}
\end{multline}
where $\mathbf{q}^{(l)}$ contains the complementary variances of the users in the $l$th iteration.
It is worth mentioning that $|\mathbf{g}_i^H\mathbf{q}+\tilde{\mathbf{f}}^H_i\mathbf{p}^{\star}|^2$ is always greater than or equal to the right-hand side of \eqref{app-sub}, and consequently, no trust region is required in DCP 
\cite{duchi2018sequential,lipp2016variations,yuille2003concave}.
Finally, in the $l$th each iteration, \eqref{59b-app} can be approximated by
\begin{multline}\label{59b-app-2}
-|\mathbf{f}_k^H\mathbf{q}+\tilde{\mathbf{f}}_k^H\mathbf{p}^{\star}|^2+E_{p,k}|\mathbf{g}_k^H\mathbf{q}^{(l)}+\tilde{\mathbf{f}}^H_k\mathbf{p}^{\star}|^2 \\+2E_{p,k}\mathfrak{R}\left((\mathbf{g}_k^H\mathbf{q}^{(l)}+\tilde{\mathbf{f}}^H_k\mathbf{p}^{\star})^H\mathbf{g}_k^H(\mathbf{q}-\mathbf{q}^{(l)})\right)\\+(\sigma^2+\mathbf{a}_k^T\mathbf{p}^{\star})^2-E_{p,k}(\sigma^2+\mathbf{b}_k^T\mathbf{p}^{\star})^2\geq \alpha_kt_k.
\end{multline}
Finally, the convex optimization problem in the $l$th iteration is 
\begin{subequations}
 \begin{align}
 \underset{t_1,t_2,\mathbf{q}}{\text{maximize}}\,\,\,\,\,\,\,\,  & 
  \min (t_1,t_2) \\
 \label{59b} \text{s.t.}  \,\,\,\,\,\,\,\, \,\,\,\,\,\,\,\,\,\,\,\,\,\,\,\,&  \eqref{59b-app-2},\eqref{ka-Const-59}.
 \end{align}
\label{ensub-itr2}
\end{subequations}
This problem can be easily solved by standard numerical tools \cite{boyd2004convex}. Moreover, the proposed DCP algorithm converges to a stationary point of \eqref{enpr-iGS-2}  \cite{duchi2018sequential,lipp2016variations,yuille2003concave,sun2017majorization,lanckriet2009convergence}.
It is worth mentioning that a stationary point of \eqref{enpr-iGS-2}  is not necessarily a stationary point of \eqref{enpr-iGS}.

The proposed simplified algorithm 
can be summarized as follows. 
The joint optimization problem for $\mathbf{p}$ and $\mathbf{q}$ is decoupled into two separate optimization problems.
We derive the transmission powers by employing the well-known bisection method, which results, in each iteration, in a feasibility problem that has a closed-form solution. 
Then, we employ the DCP algorithm to derive the complementary variances for the given transmission powers.

\section{Numerical Results}\label{secIV}
In this section, we present some numerical results to illustrate our findings. 
For all examples, we set $\sigma^2=1$,  $P_1=P_2=P$, $\epsilon=10^{-4}$, and $L=M=20$, where $\epsilon$, $L$, and $M$ are, respectively, the threshold for convergence, and the maximum number of iterations for Algorithms I and II. 
Moreover, the maximum number of iterations for the algorithm in Section \ref{Sec-cvd} is 40. 
We also define the signal-to-noise ratio (SNR) as the ratio of the power budget to $\sigma^2$, i.e., SNR$=\frac{P}{\sigma^2}$.
We compare our proposed algorithms with PGS and the joint variance and complementary variance optimization algorithm in   \cite{zeng2013transmit} for IGS, which is designed for ideal devices. 
To the best of our knowledge, there exists no PGS algorithm for additive asymmetric HWD in the literature. Because of that, we optimize the PGS scheme by using the first step of our simplified algorithm (see Section \ref{Sec-PGS}).
In the figures, we use the following labels: 
\begin{itemize}
\item {\bf S-IGS}: our proposed simplified  design in Section IV,
\item {\bf FP-IGS}: our proposed design with FP  in Section III,
\item {\bf PGS}: the proposed PGS design in Section IV-A,
\item {\bf I-IGS}: the joint variance and complementary variance IGS design  in \cite{zeng2013transmit} for ideal devices, 
\item {\bf S-TS}: our proposed design in Section IV with time sharing,
\item {\bf F-TS}: our proposed design in Section III with time sharing,
\item {\bf P-TS}: the proposed PGS design in Section IV-A with time sharing. 
\end{itemize}

\subsection{Ideal devices}
In this subsection, we compare the performance of our proposed algorithms with the joint variance and covariance IGS algorithm in \cite{zeng2013transmit} when there is no HWD. 
\begin{figure}[t]
\centering
 \includegraphics[width=0.4\textwidth]{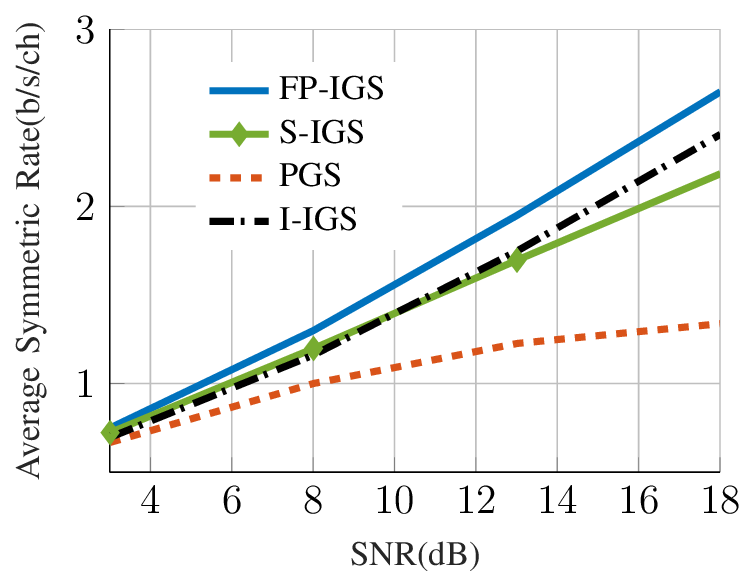}
\caption{Average symmetric rate 
for ideal devices versus  the SNR.}
\label{snr-ideal}
\end{figure}
In Fig. \ref{snr-ideal}, we show the average symmetric rate, i.e., the minimum rate allocated to the users, 
which is the fairness point of the rate region boundary and obtained by $\alpha_1=\alpha_2=0.5$. 
We average the results over 100 channel realizations, where each channel realization is taken from a complex proper Gaussian distribution with variance 1, i.e., $\mathcal{CN}(0,1,0)$.
As  can be observed, our proposed algorithm based on FP outperforms the proposed algorithm in \cite{zeng2013transmit}, especially at 
high  SNR. 
Our simplified algorithm performs similarly to the proposed algorithm in  \cite{zeng2013transmit} for low SNR. However, the algorithm in  \cite{zeng2013transmit} performs better than the simplified algorithm in the moderate SNR regime. 
The reason is that the benefit of employing IGS increases with SNR. 
Thus, the performance differences of the IGS algorithms are clearer at higher SNR.

\begin{figure*}[t!]
    \centering
    \begin{subfigure}[t]{0.5\textwidth}
        \centering
        \includegraphics[width=.8\textwidth]{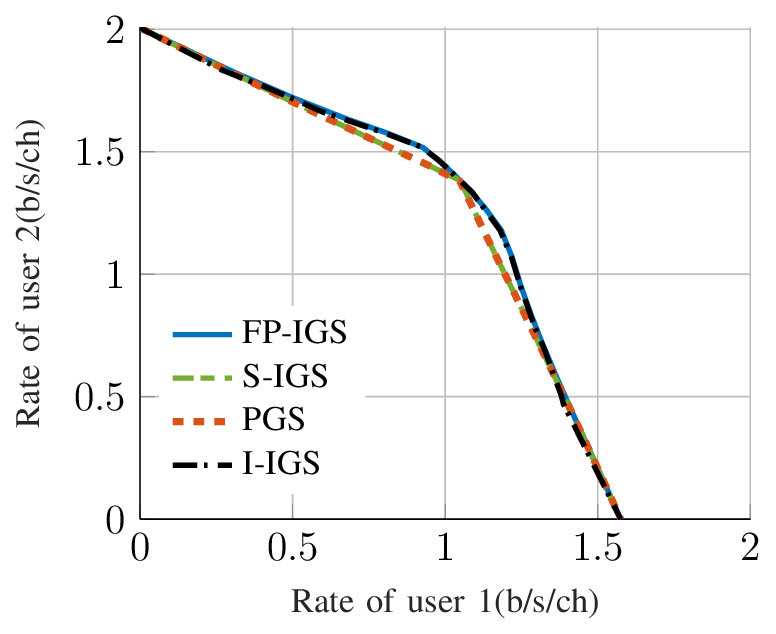}
        \caption{$P=1$ (SNR$=0\,$dB)}
    \end{subfigure}%
    ~ 
    \begin{subfigure}[t]{0.5\textwidth}
        \centering
        \includegraphics[width=.75\textwidth]{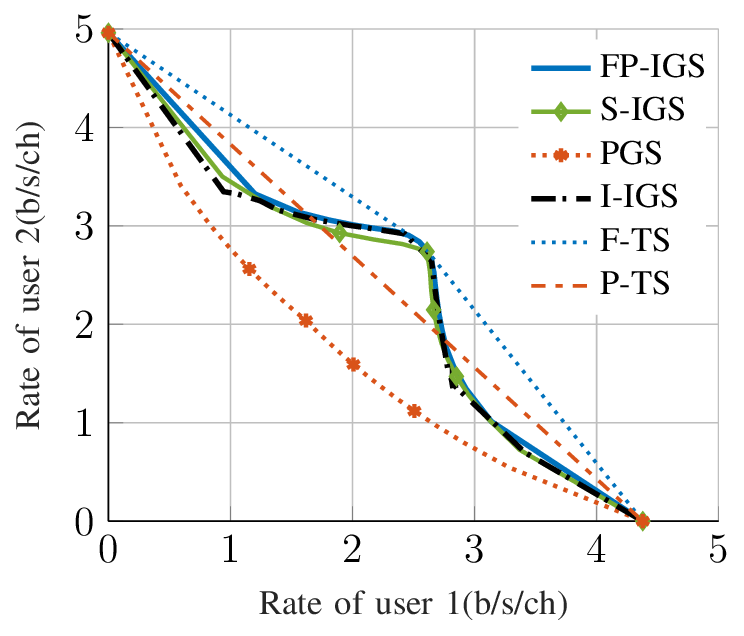}
        \caption{$P=10$ (SNR$=10\,$dB)}
    \end{subfigure}
    \caption{Achievable rate region for ideal devices and channel realization $\mathbf{H}_1$ in \eqref{h-1-ch}.}
	\label{Fig-RR-21}
\end{figure*}
In Fig. \ref{Fig-RR-21}, we also provide rate region examples for  
ideal devices and 
the channel realization
\begin{equation}\label{h-1-ch}
\mathbf{H}_1=\left[ \begin{array}{cc}
1.4070e^{i0.2721} & 0.9288e^{i1.8320 } \\ 0.9288e^{i1.8320 }  & 1.7367e^{i1.1136} \end{array} \right],
\end{equation}
where $[\mathbf{H}_1]_{ij}=h_{ij}$ for $i,j\in\{1,2\}$. 
As can be observed, IGS can enlarge the achievable rate region for this channel realization and $P=10$. 
Since the benefits of IGS are minor for low SNR, IGS does not provide any gain for $P=1$. 
This is also in line with the averaged results in Fig.  \ref{snr-ideal}, where IGS has minor benefits at low SNR, 
while it improves the performance of the system significantly at moderate SNR.
For this channel realization, our proposed algorithms and the algorithm in \cite{zeng2013transmit} perform very closely to each other. 
In Fig. \ref{Fig-RR-21}b, we also consider the effect of time sharing\footnote{
We derive the achievable rate region with TS by taking the convex hull operation over the corresponding achievable rate regions \cite{zeng2013transmit}. 
It is worth mentioning that time sharing results in the convex hull operation when power constraint is considered for each operational point. 
The achievable rate region with time sharing might be enlarged if an average power constraint over different operational point is considered  \cite{hellings2018improper}. However, this analysis is outside of the scope of this paper.  
} on the achievable rate region. As can be observed, IGS with time sharing outperforms PGS with time sharing for this example. Since the IGS designs perform similarly, for this example, we provide only the time sharing for our proposed IGS design in Section III.

The joint variance and covariance IGS algorithm in \cite{zeng2013transmit} is an iterative algorithm, based on a bisection method over the minimum weighted rates of users, and is proposed for {\em ideal devices}.
The algorithm employs semidefinite relaxation (SDR) programming in order to solve the corresponding feasibility problem at each iteration of the bisection method.  
Since the solution of the SDR in  \cite{zeng2013transmit} is not ensured to be rank-one, it does not necessarily obtain a valid solution, and a Gaussian  randomization procedure is employed to obtain a rank-one solution.  
The solution obtained by the randomization procedure is not ensured to fulfill any optimality condition, which is in contrast with our proposed algorithm, which converges to a stationary point of \eqref{enpr}.
That may be the reason why our algorithm provides a better average symmetric  rate than SDR for high SNR in this scenario.
It is also worth mentioning that our proposed algorithms are more general since they consider additive asymmetric HWD, while the algorithm in \cite{zeng2013transmit} can only be applied for ideal devices.

\begin{figure*}[t!]
    \centering
    \begin{subfigure}[t]{0.5\textwidth}
        \centering
\includegraphics[width=.8\textwidth]{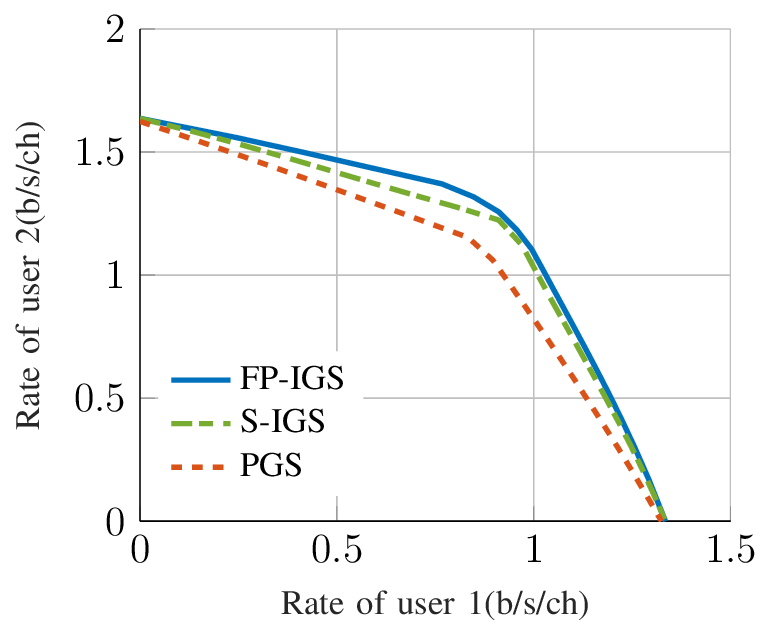}
        \caption{$\sigma^2_{\eta}=0.2$}
    \end{subfigure}%
    ~ 
    \begin{subfigure}[t]{0.5\textwidth}
        \centering
\includegraphics[width=.8\textwidth]{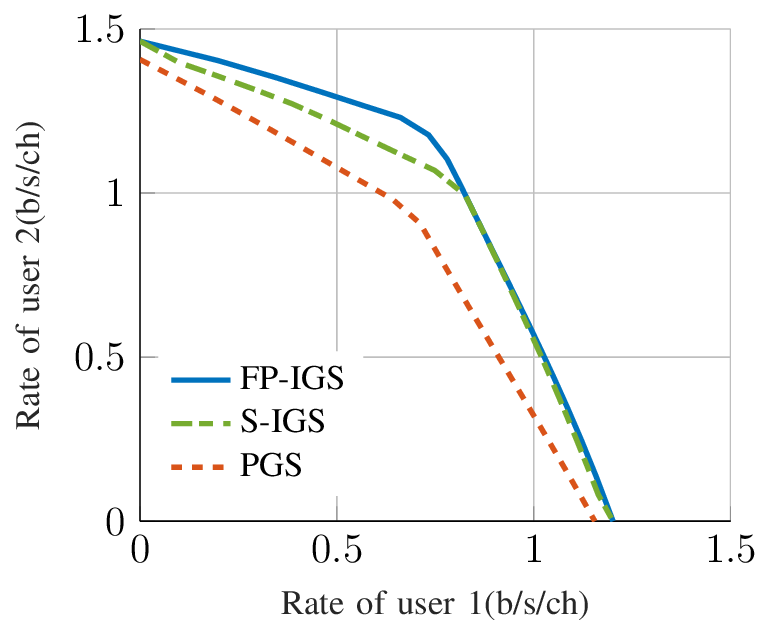}
        \caption{$\sigma^2_{\eta}=0.5$}
    \end{subfigure}
    \caption{Achievable rate region for $|\tilde{\sigma}_{\eta}^2|=\sigma^2_{\eta}$, $P=1$, 
    and channel realization $\mathbf{H}_1$ in \eqref{h-1-ch}.}
	\label{Fig-5-rr}
\end{figure*}
\subsection{Non-ideal devices}
In this subsection, we consider the effect of HWD on the performance of the two-user IC. 
Throughout this subsection, we consider the same statistics for HWD in all devices. 
In Fig. \ref{Fig-5-rr}, we show the rate region for $\mathbf{H}_1$ and $P=1$ under maximally improper HWD\footnote{Maximally improper HWD happens when the in-phase and quadrature-phase noises are completely correlated \cite{javed2018improper}.} noise. 
As shown in Fig. \ref{Fig-RR-21}a, IGS brings negligible gains when the transceivers are ideal, but, as observed in Fig. \ref{Fig-5-rr}, IGS can significantly enlarge the rate region if there is additive asymmetric HWD.
Note that even in point-to-point communications, PGS is in general suboptimal for asymmetric HWD, as it is shown in Fig. \ref{Fig-5-rr} for either $R_1=0$ or $R_2=0$.

\begin{figure*}[t!]
    \centering
    \begin{subfigure}[t]{0.5\textwidth}
        \centering
\includegraphics[width=.8\textwidth]{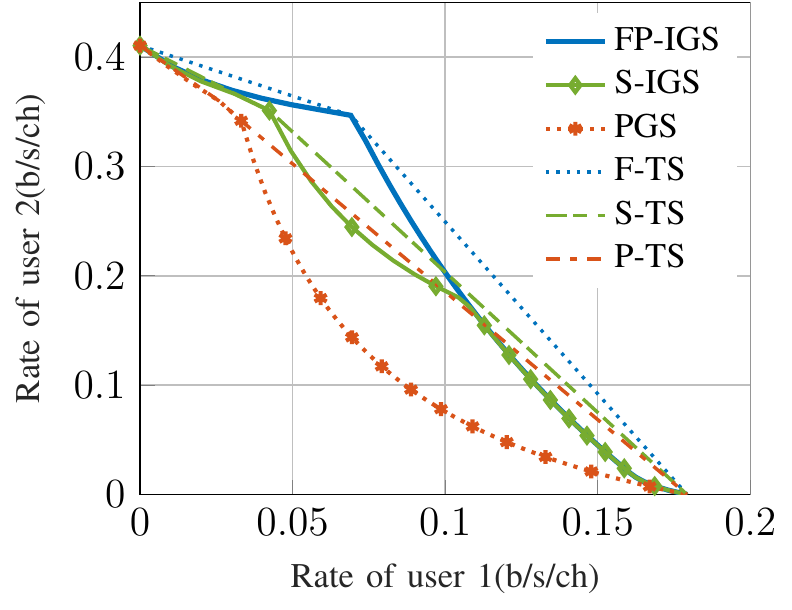}
        \caption{$\sigma^2_{\eta}=0.5$}
    \end{subfigure}%
    ~ 
    \begin{subfigure}[t]{0.5\textwidth}
        \centering
\includegraphics[width=.8\textwidth]{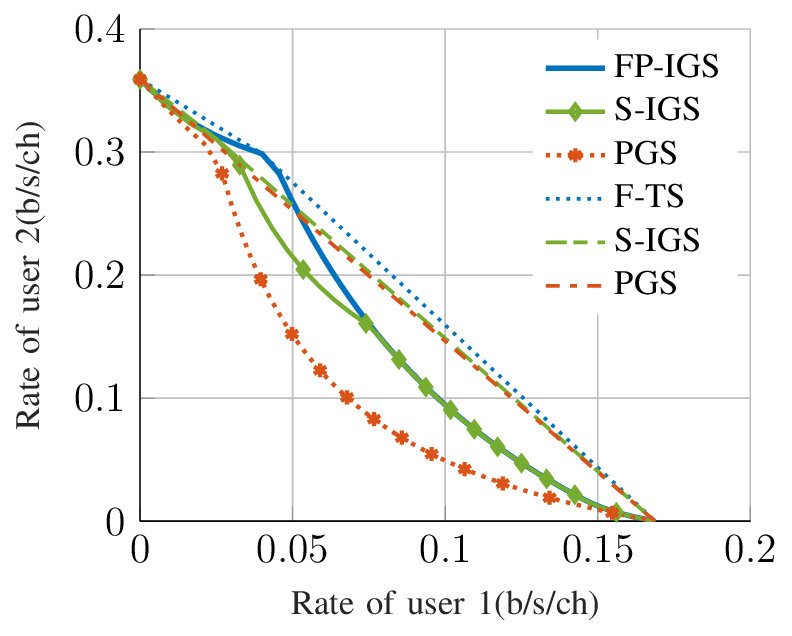}
        \caption{$\sigma^2_{\eta}=1$}
    \end{subfigure}
    \caption{Achievable rate region for $\tilde{\sigma}_{\eta}^2=0$, $P=1$, 
    and channel realization $\mathbf{H}_2$ in \eqref{h-2-ch}.}
	\label{Fig-RR-p-2}
\end{figure*}
In Fig. \ref{Fig-RR-p-2}, we show the achievable rate region for $\tilde{\sigma}_{\eta}^2=0$, $P=1$ (SNR$=0\,$dB),  
and channel realization
\begin{equation}\label{h-2-ch}
\mathbf{H}_2=\left[ \begin{array}{cc}
0.3764 e^{i1.4381} & 0.4029e^{i0.9486} \\ 1.8542e^{i2.8153}  & 0.6277e^{i2.3697} \end{array} \right].
\end{equation}
We take $\tilde{\sigma}^2_{\eta}=0$, i.e., symmetric (proper) HWD. We can observe that IGS enlarges the rate region even for proper  HWD with high noise variance, i.e., $\sigma^2_{\eta}=0.5$ and $\sigma^2_{\eta}=1$. 
It is worth mentioning that the PGS design is Pareto-optimal in the presence of additive symmetric HWD. 
As can be observed, our IGS design in Section III with time sharing outperforms the Pareto-optimal PGS with time sharing for these examples.  

In the following, we provide some averaged results for different parameters to illustrate different aspects of employing IGS.
Similar to Fig. \ref{snr-ideal}, we average the results over 100 channel realizations, 
where each channel realization is taken from a complex proper Gaussian distribution with variance 1, i.e., $\mathcal{CN}(0,1,0)$.

\begin{figure*}[t!]
    \centering
    \begin{subfigure}[t]{0.5\textwidth}
        \centering
\includegraphics[width=.8\textwidth]{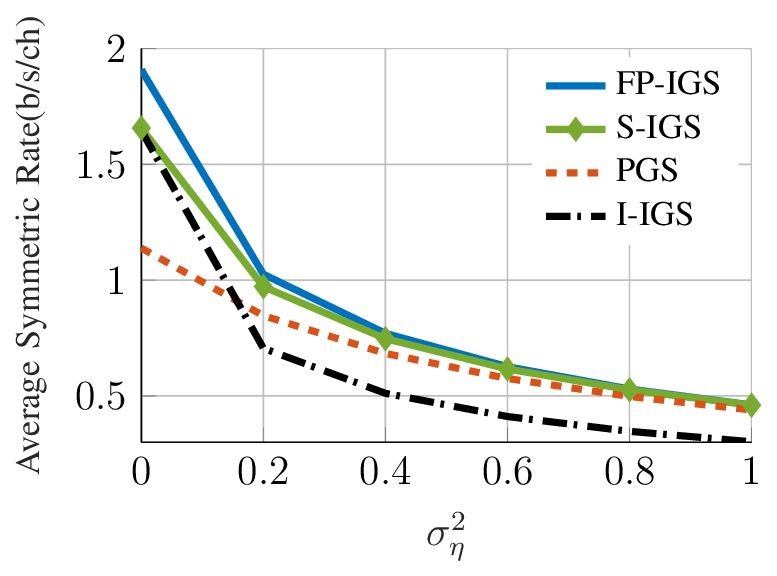}
        \caption{$\tilde{\sigma}^2_{\eta}=0$}
    \end{subfigure}%
    ~ 
    \begin{subfigure}[t]{0.5\textwidth}
        \centering
\includegraphics[width=.8\textwidth]{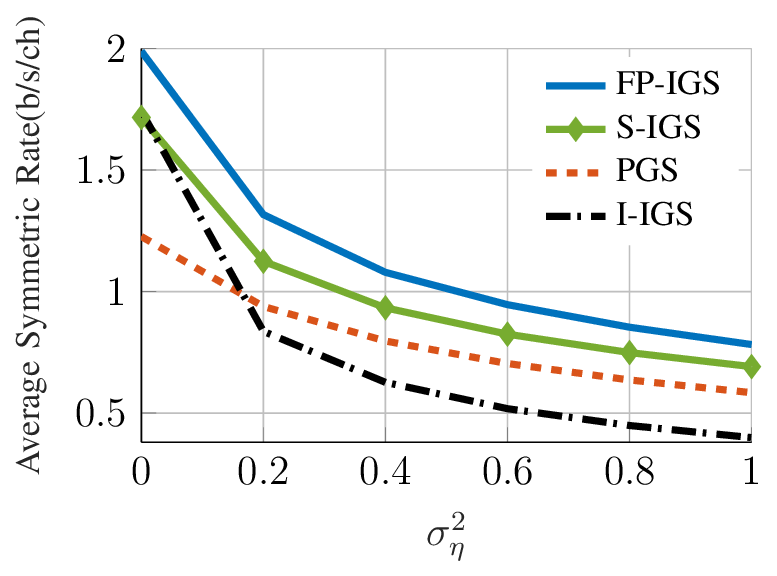}
        \caption{$\tilde{\sigma}^2_{\eta}=\sigma^2_{\eta}$.}
    \end{subfigure}
    \caption{Average symmetric rate 
versus the variance of  the HWD noise for  $P=20$.}
	\label{Fig2}
\end{figure*}
In Fig. \ref{Fig2}, we consider the effect of the variance of the HWD noise on the average symmetric rate of users ($\alpha_1=\alpha_2=0.5$) for  
$P=20$. 
In this figure, we consider proper ($\tilde{\sigma}^2_{\eta}=0$) and maximally improper ($\tilde{\sigma}^2_{\eta}=\sigma^2_{\eta}$)  HWD noise. 
We observe that our proposed algorithm with FP outperforms the other algorithms for maximally improper HWD noise. 
Moreover, in Fig. \ref{Fig2}a, our proposed IGS algorithms perform better than PGS for proper HWD noise with different variances, which is Pareto-optimal PGS in this case. 
Furthermore, our simplified algorithm outperforms the IGS algorithm in \cite{zeng2013transmit} in the presence of HWD. 
However, the performance improvement by our algorithms is minor for proper HWD with high noise variance, where our algorithms only provide $5\%$ improvement over PGS when $\sigma^2_{\eta}=1$ for this example.

\begin{figure*}[t!]
    \centering
    \begin{subfigure}[t]{0.5\textwidth}
        \centering
\includegraphics[width=.8\textwidth]{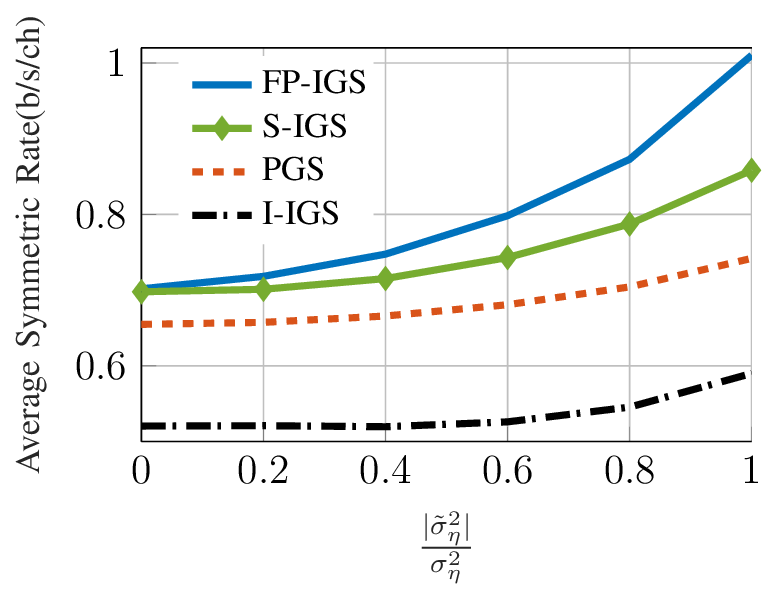}
        \caption{$\sigma_{\eta}^2=0.5$}
    \end{subfigure}%
    ~ 
    \begin{subfigure}[t]{0.5\textwidth}
        \centering
\includegraphics[width=.8\textwidth]{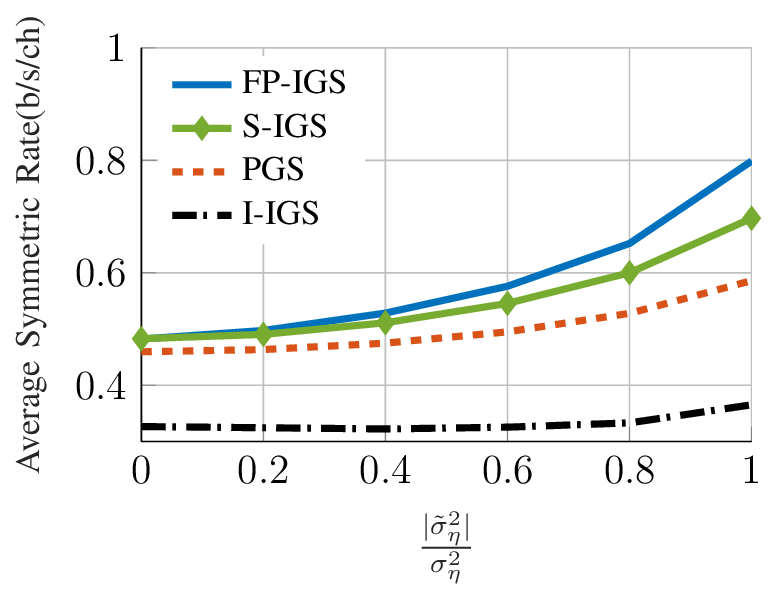}
        \caption{$\sigma_{\eta}^2=1$}
    \end{subfigure}
    \caption{Average symmetric rates 
    versus  the circularity coefficient of HWD noise for 
    $P=20$.}
	\label{Fig3}
\end{figure*}
Figure \ref{Fig3} shows the effect of the circularity coefficient of the HWD noise on the symmetric rate 
for 
$P=20$.  
As can be observed, 
the benefits of employing IGS increase with the circularity coefficient of the HWD noise, and there is a considerable performance improvement by IGS in maximally improper HWD noise. 
We emphasize that PGS is suboptimal, even in interference-free communications, under asymmetric HWD. 
Our proposed IGS design with FP outperforms the other algorithms, especially in highly asymmetric HWD noise.
When the variance of the HWD noise is small, the gain of employing IGS is larger. 
The other interesting result in this figure is that our simplified algorithm performs very similarly to our proposed algorithm based on FP for proper HWD. 
Since the simplified algorithm has less computational cost, it can be employed for proper HWD noise when the variance of the  HWD noise is high, i.e., $\sigma^2_{\eta}\geq 0.5$. 
However,  our proposed algorithm based on FP outperforms the other algorithm in low-power HWD noise and/or highly asymmetric HWD noise.
Note that, since the IGS algorithm in \cite{zeng2013transmit} is proposed for ideal devices and does not consider HWD, it performs worse than the proposed PGS, which considers additive symmetric HWD, from the average symmetric rate point of view, even when the HWD noise is maximally improper.

\begin{figure*}[t!]
    \centering
    \begin{subfigure}[t]{0.5\textwidth}
        \centering
\includegraphics[width=.8\textwidth]{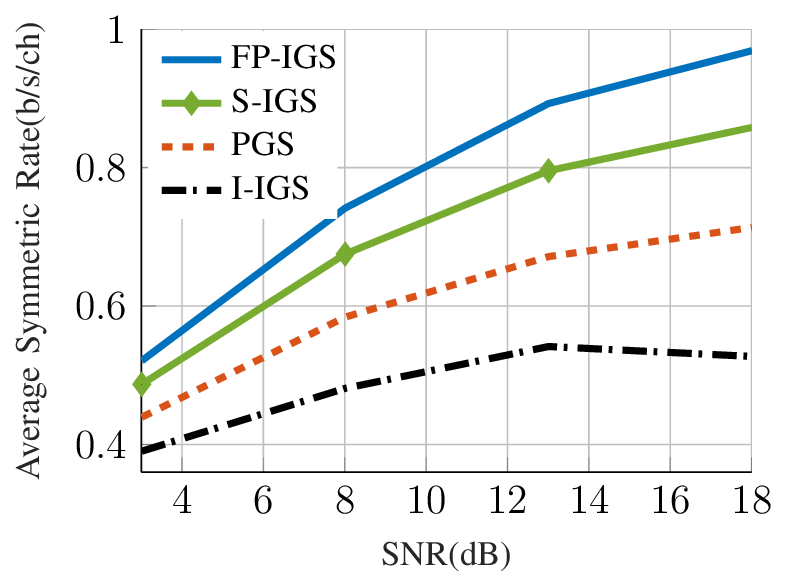}
        \caption{$\sigma_{\eta}^2=0.5$}
    \end{subfigure}%
    ~ 
    \begin{subfigure}[t]{0.5\textwidth}
        \centering
\includegraphics[width=.8\textwidth]{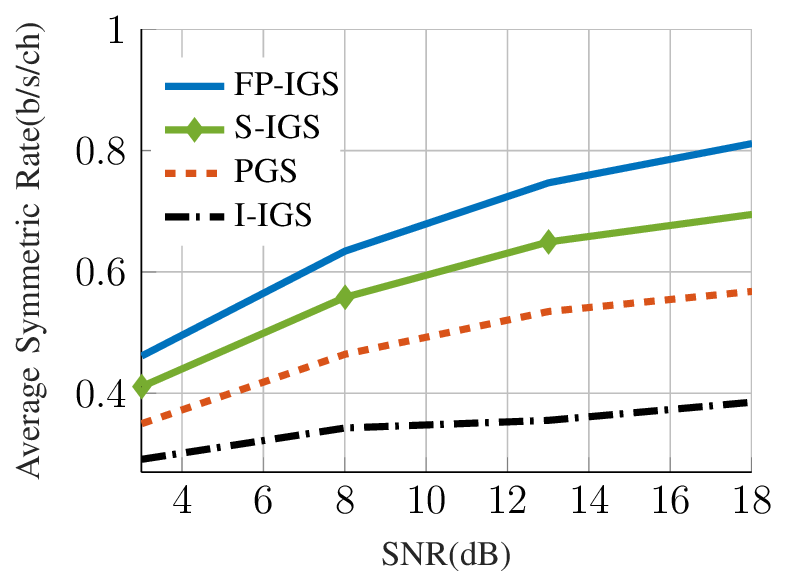}
        \caption{$\sigma_{\eta}^2=1$}
    \end{subfigure}
    \caption{Average symmetric rate 
    versus SNR 
    for $\tilde{\sigma}^2_{\eta}=0.9\sigma^2_{\eta}$.}
	\label{Fig4}
\end{figure*}

In Fig. \ref{Fig4}, we consider the effect of the power budget on the symmetric rate of users. 
There is an almost constant performance gap between our proposed algorithms and the other algorithms. Similar to the other figures, our proposed IGS with FP outperforms our simplified algorithm.

\section{Conclusion}

In this paper, we considered a two-user IC with additive asymmetric HWD at the transceivers. 
Treating interference as noise, we addressed the problem of obtaining the achievable rate region for IGS and proposed two suboptimal algorithms.
The first algorithm, which is based on MM and the generalized Dinkelbach algorithm, obtains a stationary point of the PSINR region. 
In this algorithm, we jointly optimize the powers and complementary variances. 
We also proposed a simplified algorithm that has lower computational complexity. 
This simplified algorithm is based on the separate optimization of the powers and complementary variances.
Through numerical examples, we showed that the proposed approaches enlarge the achievable rate region and outperform PGS and existing IGS algorithms, especially as the HWD becomes more asymmetric.


\appendices
\section{Proof of Lemma \ref{uvfunk-app}}\label{app-new}
In order to approximate $u_k(\mathbf{p},\mathbf{q})$ and $v_k(\mathbf{p},\mathbf{q})$, we employ convex-concave (or concave-convex) procedure (CCP), in which the convex (concave) part is approximated as an affine function by
the first-order approximation of the Taylor expansion.
Note that we take the first-order term and employ an affine approximation since an affine function is the nearest concave approximation to a convex function.
The first-order approximation of a real function $u(\mathbf{x})$ around the point $\mathbf{x}_0$ is obtained through its Taylor expansion as \cite{ schreier2010statistical, brandwood1983complex}
\begin{equation}\label{Taylor}
\Gamma(\mathbf{x})\approx \Gamma(\mathbf{x}_0)+
2\mathfrak{R}\left[(\frac{\partial \Gamma(\mathbf{x})}{\partial \mathbf{x}}\Big|_{\mathbf{x}=\mathbf{x}_0})^T(\mathbf{x}-\mathbf{x}_0)\right],
\end{equation}
where $\mathbf{x}$ is a complex vector.  
In order to apply the CCP to $u_k(\mathbf{p},\mathbf{q})$, we have to differentiate the convex part in \eqref{ufunk} with respect to $\mathbf{p}$, which 
is straightforward since it is a real function on a real domain and consequently, analytic in $\mathbf{p}$. 
The derivative of $(\sigma^2+\mathbf{a}_k^T\mathbf{p})^2$ with respect to $\mathbf{p}$ is 
\begin{align}
\frac{\partial(\sigma^2+\mathbf{a}_k^T\mathbf{p})^2}{\partial\mathbf{p}}&=2\mathbf{a}_k(\sigma^2+\mathbf{a}_k^T\mathbf{p}),
\end{align}
and the resulting first-order approximation around the power vector in the $m$th iteration, $\mathbf{p}^{(m)}$, is given by
\begin{multline}
\label{siktir}
(\sigma^2+\mathbf{a}_k^T\mathbf{p})^2\simeq (\sigma^2+\mathbf{a}_k^T\mathbf{p}^{(m)})^2\\+2(\sigma^2+\mathbf{a}_k^T\mathbf{p}^{(m)})\mathbf{a}_k^T(\mathbf{p}-\mathbf{p}^{(m)}).
\end{multline}
By substituting \eqref{siktir} in \eqref{ufunk}, we can derive $\tilde{u}_k(\mathbf{p},\mathbf{q})$.

 In order to convexify $v_k(\mathbf{p},\mathbf{q})$, we have to differentiate the concave part in \eqref{vfunk} with respect to $\mathbf{p}$ and $\mathbf{q}$.
The derivative of $|\mathbf{g}_k^H\mathbf{q}+\tilde{\mathbf{f}}^H_k\mathbf{p}|^2$ with respect to $\mathbf{p}$ is also straightforward since it is analytic in $\mathbf{p}$:
\begin{equation}
\frac{\partial|\mathbf{g}_k^H\mathbf{q}+\tilde{\mathbf{f}}^H_k\mathbf{p}|^2}{\partial\mathbf{p}}=2\mathfrak{R}\left[\tilde{\mathbf{f}}_k(\mathbf{g}_k^H\mathbf{q}+\tilde{\mathbf{f}}^H_k\mathbf{p})\right].
\end{equation}
The term $|\mathbf{g}_k^H\mathbf{q}+\tilde{\mathbf{f}}^H_k\mathbf{p}|^2$, on the other hand, is not analytic in $\mathbf{q}$ since it is a real-valued function while $\mathbf{q}$ is a complex vector \cite{schreier2010statistical,brandwood1983complex}. 
Thus, we have to employ Wirtinger calculus to obtain the derivative of $|\mathbf{g}_k^H\mathbf{q}+\tilde{\mathbf{f}}^H_k\mathbf{p}|^2$ with respect to $\mathbf{q}$. 
By Wirtinger calculus, we treat $\mathbf{q}$ and $\mathbf{q}^*$ as two independent complex variables \cite{schreier2010statistical,brandwood1983complex}. 
Thus,  we take the derivative of $|\mathbf{g}_k^H\mathbf{q}+\tilde{\mathbf{f}}^H_k\mathbf{p}|^2$ with respect to $\mathbf{q}$ while treating $\mathbf{q}^*$ as a constant, which results in 
\begin{align}
\frac{\partial|\mathbf{g}_k^H\mathbf{q}+\tilde{\mathbf{f}}^H_k\mathbf{p}|^2}{\partial\mathbf{q}}&=\mathbf{g}_k^*(\mathbf{g}_k^H\mathbf{q}+\tilde{\mathbf{f}}^H_k\mathbf{p})^*. 
\end{align}
Now by  \eqref{Taylor}, we can approximate $|\mathbf{g}_k^H\mathbf{q}+\tilde{\mathbf{f}}^H_k\mathbf{p}|^2$ as an affine function as 
\begin{multline}\label{siktir-1}
|\mathbf{g}_k^H\mathbf{q}+\tilde{\mathbf{f}}^H_k\mathbf{p}|^2\simeq |\mathbf{g}_k^H\mathbf{q}^{(m)}+\tilde{\mathbf{f}}^H_k\mathbf{p}^{(m)}|^2\\ +2\mathfrak{R}\left(\tilde{\mathbf{f}}_k^H(\mathbf{g}_k^H\mathbf{q}^{(m)} +\tilde{\mathbf{f}}^H_k\mathbf{p}^{(m)})^*\right)(\mathbf{p}-\mathbf{p}^{(m)})\\
+2\mathfrak{R}\left((\mathbf{g}_k^H\mathbf{q}^{(m)}+\tilde{\mathbf{f}}^H_k\mathbf{p}^{(m)})^*\mathbf{g}_k^H(\mathbf{q}-\mathbf{q}^{(m)})\right).
\end{multline}
By substituting \eqref{siktir-1} in \eqref{vfunk}, we can obtain $\tilde{v}_k(\mathbf{p},\mathbf{q})$.

\section{Proof of Theorem \ref{E-lower}}\label{app:th-1}

\begin{figure*}[t!]
    \centering
    \begin{subfigure}[t]{0.5\textwidth}
        \centering
       \includegraphics[width=0.7\textwidth]{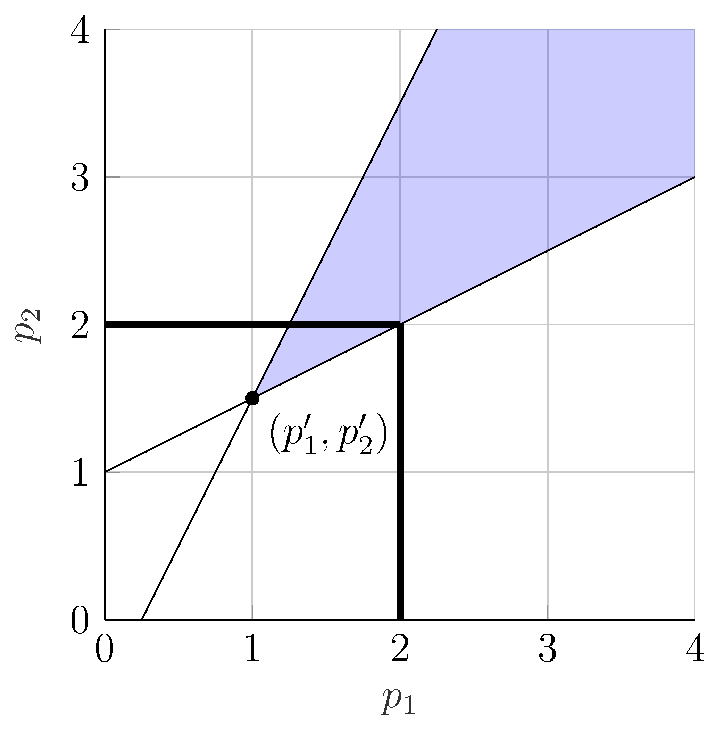}
        \caption{Feasible $E^{\prime}$}
    \end{subfigure}%
    ~ 
    \begin{subfigure}[t]{0.5\textwidth}
        \centering
        \includegraphics[width=0.7\textwidth]{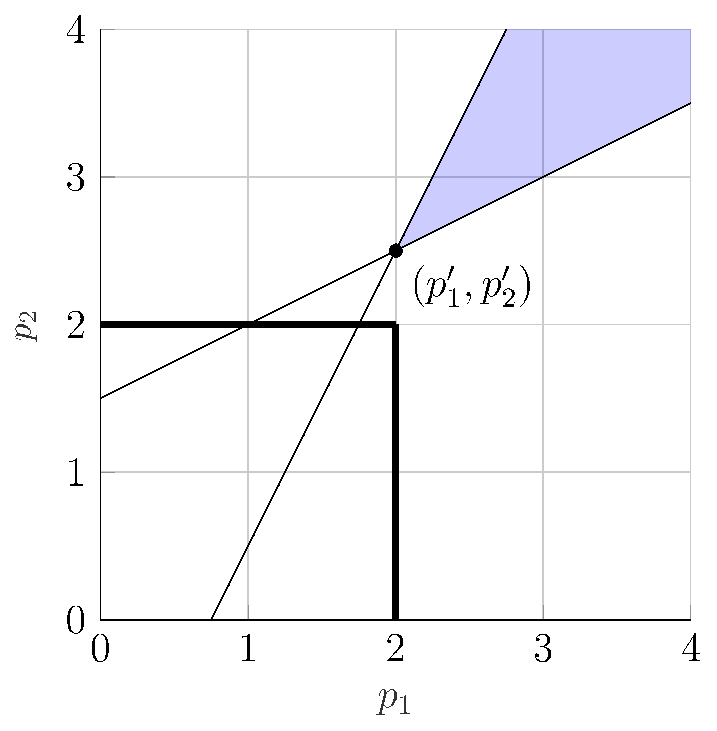}
        \caption{Infeasible $E^{\prime}$}
    \end{subfigure}
    \caption{The constraints of \eqref{enpr-PGS-app-bi} in the power plane.}
	\label{intersec}
\end{figure*}
A given $E^{\prime}$ is feasible if and only if there exists at least a pair $(p_1,p_2)$ that satisfies all the constraints in \eqref{enpr-PGS-app-bi}. 
Let us first  consider the two linear constraints in \eqref{e-Const-40}, which can be written as \eqref{shabes-1} and \eqref{shabes-2}, shown at the top of the next page.
\begin{figure*}
\hrulefill
\begin{align}
|h_{11}|^2\!\left(\!\!1\!-\!\sigma^2_{\eta_{11}}\!(\!\sqrt{1\!+\!\alpha_1 E^{\prime}}\!-\!1\!)\!\right)\!p_1\!-\!|h_{21}|^2(1\!+\!\sigma^2_{\eta_{21}})(\!\sqrt{1\!+\!\alpha_1 E^{\prime}}\!-1\!)p_2\!
\geq\! (\!\sqrt{\!1\!+\!\alpha_1 E^{\prime}}\!-\!1)\sigma^2\!\!,\label{shabes-1}\\
-|h_{12}|^2\!(\!1\!+\!\sigma^2_{\eta_{12}})(\!\sqrt{1\!+\!\alpha_2 E^{\prime}}\!-\!1\!)p_1+|h_{22}|^2\!\left(\!1\!-\!\sigma^2_{\eta_{22}}(\!\sqrt{\!1\!+\!\alpha_2 E^{\prime}}\!-\!1\!)\!\right)\!p_2\!
\geq\! (\!\sqrt{\!1\!+\!\alpha_2 E^{\prime}}\!-\!1\!)\sigma^2\!\!.\label{shabes-2}
\end{align}
\hrulefill
\end{figure*}
We can construct $\mathbf{A}$ in \eqref{line-slop} by the coefficients of $p_1$ and $p_2$ in \eqref{shabes-1} and \eqref{shabes-2}. 
It is worth mentioning that the non-diagonal elements of $\mathbf{A}$ in \eqref{line-slop} are non-positive since $\sqrt{1+\alpha_1 E^{\prime}}\geq 1$ and $\sqrt{1+\alpha_2 E^{\prime}}\geq 1$.  
Thus, if the diagonal elements of $\mathbf{A}$ are not  positive, there is no positive power pair that satisfies \eqref{shabes-1} and \eqref{shabes-2} simultaneously. Hence, in the following, we assume without loss of generality that $\mathbf{A}$ has strictly positive diagonal elements and strictly negative non-diagonal elements.

We can rewrite \eqref{shabes-1} and \eqref{shabes-2}  as
\begin{align}\label{shabes-1-2}
[\mathbf{A}]_{11}p_1&\geq -[\mathbf{A}]_{12}p_2+y_1,\\
\label{shabes-2-2}
[\mathbf{A}]_{22}p_2&\geq -[\mathbf{A}]_{21}p_1+y_2,
\end{align}
where $\mathbf{y}=\left[ \begin{array}{cc}
(\sqrt{1+\alpha_1 E^{\prime}}-1)\sigma^2 & (\sqrt{1+\alpha_2 E^{\prime}}-1)\sigma^2 \end{array} \right]^T$. Moreover, $[\mathbf{A}]_{ij}$, and $y_i$ for $i,j\in\{1,2\}$ are the $ij$th element of $\mathbf{A}$, and the $i$th element of $\mathbf{y}$, respectively. If we decouple the inequalities, we end up with
\begin{align}\label{shabes-1-2-1}
\det (\mathbf{A})p_1&\geq [\mathbf{A}]_{22}y_1-[\mathbf{A}]_{12}y_2,\\
\label{shabes-2-2}
\det (\mathbf{A})p_2&\geq -[\mathbf{A}]_{21}y_1+[\mathbf{A}]_{11}y_2.
\end{align}
The right-hand sides (RHS) in \eqref{shabes-1-2-1} and \eqref{shabes-2-2} are positive for a feasible $E^{\prime}$ as mentioned before.
Note that if $\det (\mathbf{A})< 0$, there are no positive power pairs that satisfy \eqref{shabes-1-2-1} and \eqref{shabes-2-2} for the given structure of $\mathbf{A}$ in \eqref{line-slop}.  
Thus, we consider $\det (\mathbf{A})> 0$, which yields 
\begin{align}\label{shabes-1-2-1-2}
p_1&\geq p_1^{\prime}=\frac{[\mathbf{A}]_{22}y_1-[\mathbf{A}]_{12}y_2}{\det (\mathbf{A})},\\
\label{shabes-2-2-3}
p_2&\geq p_2^{\prime}=\frac{-[\mathbf{A}]_{21}y_1+[\mathbf{A}]_{11}y_2}{\det (\mathbf{A})},
\end{align}
or equivalently $\mathbf{p}=[\begin{array}{cc}p_1 &p_2\end{array}]^T\geq \mathbf{A}^{-1}\mathbf{y}$, where $p_1^{\prime}$ and $p_2^{\prime}$ are the intersecting point given in \eqref{int-sec-pnt}.
Hence, the intersecting point provides the minimum positive power pairs that satisfy \eqref{shabes-1} and \eqref{shabes-2}. 
If $p_1^{\prime}$ and $p_2^{\prime}$ satisfy the power constraint, $E^{\prime}$ is feasible (Fig. \ref{intersec}.a). Otherwise, $E^{\prime}$ is infeasible (Fig. \ref{intersec}.b).

\bibliographystyle{IEEEtran}
\bibliography{ref2}
\end{document}